\documentclass[aps,twocolumn,superscriptaddress,notitlepage,pra,showpacs,nofootinbib]{revtex4-1}
\usepackage[english]{babel}
\usepackage{amsmath}
\usepackage{amssymb}
\usepackage{amsthm}
\usepackage{bbm}
\usepackage{graphicx}
\usepackage{multirow}
\usepackage{color,soul,changes}

\newtheorem{theorem}{Theorem}
\newtheorem{lemma}[theorem]{Lemma}
\newtheorem{proposition}{Proposition}
\newtheorem{definition}{Definition}

\usepackage{tikz}
\def\checkmark{\tikz\fill[scale=0.4](0,.35) -- (.25,0) -- (1,.7) -- (.25,.15) -- cycle;} 

\usepackage{pifont}
\newcommand{\xmark}{\ding{55}}%

\newcommand{\ket}[1]{| #1 \rangle}
\newcommand{\proj}[1]{| #1 \rangle\!\langle #1 |}
\renewcommand{\top}{^\text{\tiny T}}

\newcommand{\C}{\mathcal{C}}
\newcommand{\F}{\mathcal{F}}
\newcommand{\I}{\mathcal{I}}
\newcommand{\X}{\mathcal{X}}
\newcommand{\Y}{\mathcal{Y}}

\newcommand{\N}{\mathcal{N}}
\newcommand{\NSspace}{\ensuremath{\overline{\mathcal{N}}}} 
\newcommand{\Q}{\mathcal{Q}}
\renewcommand{\S}{\mathcal{S}}
\newcommand{\AQ}{\tilde{\Q}}
\newcommand{\NSset}{\ensuremath{\mathcal{N}}} 
\newcommand{\Pset}{\ensuremath{\mathcal{P}}} 
\newcommand{\Qset}{\mathcal{Q}}

\newcommand{\CHSH}{\text{\tiny CHSH}}

\newcommand{\Nt}{N_\text{\tiny trials}}
\newcommand{\DKL}{D_\text{\tiny KL}}

\newcommand{\vecP}{\vec{P}}

\newcommand{\vecPQ}{\vecP_\Q}

\newcommand{\vecf}{\vec{f}}
\newcommand{\vecPpi}{\vec{P}_{\Pi}(\vecf)}
\newcommand{\vecPnqa}{\vec{P}_{{\rm LS}}({\vecf})}
\newcommand{\vecPkl}{\vec{P}_{\text{\tiny ML}}({\vecf})}

\newcommand{\vecPreg}{\vec{P}_\text{\tiny Reg}}
\newcommand{\vecPregM}[1]{\vec{P}_{\mbox{\tiny #1}}}

\newcommand{\tr}{\ensuremath{\operatorname{tr}}}
\newcommand{\one}{\mathbb{I}}

\DeclareMathOperator*{\argmin}{argmin}
\DeclareMathOperator*{\argmax}{argmax}

\usepackage{hyperref}

\renewcommand{\L}{\mathcal{L}}

\begin{document}

\title{Device-independent point estimation from finite data\\ and its application to device-independent property estimation}

\author{Pei-Sheng Lin}%
\affiliation{Department of Physics, National Cheng Kung University, Tainan 701, Taiwan}
\author{Denis Rosset}
\affiliation{Department of Physics, National Cheng Kung University, Tainan 701, Taiwan}
\author{Yanbao Zhang}
\affiliation{NTT Basic Research Laboratories, NTT Corporation, 3-1 Morinosato-Wakamiya, Atsugi, Kanagawa 243-0198, Japan}
\affiliation{Institute for Quantum Computing, University of Waterloo, Waterloo, Ontario N2L 3G1, Canada}
\affiliation{Department of Physics and Astronomy, University of Waterloo, Waterloo, Ontario N2L 3G1, Canada}
\author{Jean-Daniel Bancal}
\affiliation{Department of Physics, University of Basel, Klingelbergstrasse 82, 4056 Basel, Switzerland}
\author{Yeong-Cherng Liang}
\email{ycliang@mail.ncku.edu.tw}
\affiliation{Department of Physics, National Cheng Kung University, Tainan 701, Taiwan}

\begin{abstract}
The device-independent approach to physics is one where conclusions are drawn directly from the observed correlations between measurement outcomes. In quantum information, this approach allows one to make strong statements about the properties of the underlying systems or devices solely via the observation of Bell-inequality-violating correlations. However, since one can only perform a {\em finite number} of experimental trials, statistical fluctuations necessarily accompany any estimation of these correlations. Consequently, an important gap remains between the many theoretical tools developed for the asymptotic scenario and the experimentally obtained raw data. In particular, a physical and concurrently practical way to estimate the underlying quantum distribution has so far remained elusive. Here, we show that the natural analogs of the maximum-likelihood estimation technique and the least-square-error estimation technique in the device-independent context result in point estimates of the true distribution that are physical, unique, computationally tractable and consistent. They thus serve as sound algorithmic tools allowing one to bridge the aforementioned gap.  As an application, we demonstrate how such estimates of the underlying quantum distribution can be used to provide, in certain cases, trustworthy estimates of  the amount of entanglement present in the measured system. In stark contrast to existing approaches to device-independent parameter estimations, our estimation does not require the prior knowledge of {\em any} Bell inequality tailored for the specific property and the specific distribution of interest. 
\end{abstract}

\date{\today}

\maketitle
\section{Introduction}
The proper analysis of empirical data is an indispensable part in the development of both science and technologies. In quantum information, for instance, the careful preparation followed by the proper characterization of quantum systems (which includes estimating reliably the prepared quantum state or the confirmation that it possesses certain desired properties) is often the first step to many quantum information processing protocols. 

In practice, however, the execution of this preliminary task is far from trivial. For example, systematic uncertainties arising from various imperfections in the setup may compromise the reliability of the estimate~\cite{Rosset:2012,Moroder2013b,vanEnk:2013}. Moreover, unavoidable statistical fluctuations  result in situations where the ideal, theoretical description become inapplicable. Hence, quantum state estimation~\cite{Hradil1997} using real data is a daunting task~\cite{Blume-Kohout2010,arXiv:1202.5270,Christandl2012, Sugiyama2013,Shang2013,Faist2016} where there remains an ongoing debate on the preferred approach (see, e.g.,~\cite{arXiv:1405.5350:brief, Schwemmer2015} and references therein).

Interestingly, the first of these problems can be circumvented, to some extent, by the so-called  {\em device-independent} approach~\cite{Scarani2012,Brunner:RMP}. There, the nature of the devices employed is deduced directly from the measurement statistics~\cite{Mayers1998,Mayers2004}, without relying on any assumption about the devices' detailed functioning or the associated Hilbert space dimension. Consequently, robust characterizations of quantum systems and instruments are now in principle possible with minimal assumptions. Likewise, the distribution of shared secret keys~\cite{Ekert1991,Acin2007,Vazirani2014} and the generation of random bits---secured by the laws of physics~\cite{ColbeckPhD,Pironio:2010aa,Colbeck2011}---are now possibilities at our disposal.

Crucially, in order to make nontrivial statements from the empirical data, the latter have to be Bell-inequality~\cite{Bell1964} violating, which cannot arise from local measurements on a separable quantum state~\cite{Werner1989}. Indeed, the extent of the observed Bell-inequality violation can be used to provide an estimate, e.g., of the amount of shared entanglement~\cite{Moroder2013,Toth2015}, or even the incompatibility between the measurements~\cite{Chen2016,Cavalcanti2016} employed. Nonetheless, as with the case where the measurement devices are fully characterized, the underlying distribution---which serves as the analog of a quantum {\em state} in this black-box setting---contains all the available information and thus generally provides a much better estimate of the system's properties~\cite{Zhang2011,Zhang2013,Bancal2014,Nieto2014}.

Indeed, various theoretical tools taking into account the full quantum distribution have been developed for device-independent characterizations: from the nature of the (multipartite) entanglement~\cite{Bancal2011,LiangPRL2015,Baccari2017} present to their quantification~\cite{Moroder2013,Toth2015}, from the steerability~\cite{Wiseman2007} of the underlying state to the incompatibility of the measurements employed~\cite{Chen2016,Cavalcanti2016}, and from the minimal compatible Hilbert space dimension~\cite{Brunner2008,NavascuesPRX,NavascuesPRL2015} to the self-testing~\cite{Mayers1998,Mayers2004} of the quantum apparatus~\cite{Yang2014,BancalPRA2015}. Stemming from the algorithmic characterization of the set of quantum distributions due to Navascu\'es-Pironio-Ac\'in (NPA)~\cite{Navascues2007,Navascues2008a}, they share the common assumption that the estimated distributions satisfy the physically motivated conditions of nonsignaling~\cite{Popescu1994,Barrett2005}. 

However, raw distributions estimated from the relative frequencies of experimental outcomes---due to statistical fluctuations---generically do not satisfy these conditions. As such, {\em none} of the aforementioned tools can be directly applied to experimentally observed statistics. In other words, while the device-independent approach offers an elegant solution to overcome the problem of mistrusting the measurement devices, there remains an important gap between the  theoretical tools developed for such purposes and the actual data available from any Bell experiment.

For the very {\em specific} problem of device-independent randomness certification and quantum key distribution, techniques based on hypothesis testing have been shown, respectively, in~\cite{Pironio:2010aa,Pironio2013,arXiv:1611.00352:brief,arXiv:1702.05178:brief,Knill2017} and  in~\cite{Dupuis:1607.01796:brief,Rotem2017} to be applicable even in the presence of finite statistics. These approaches are, however, very problem specific, and it is not yet known how to generalize them for the {\em general problem} of device-independent characterizations. Here, we consider an alternative approach inspired by estimation theory, which consists in constructing a point estimate for the underlying quantum distribution from the observed frequencies.  In particular, we show that the natural analogs of two physical estimators employed in usual quantum state tomography, namely,  maximum-likelihood (ML) estimation and least-square-error estimation, also serve as sound estimators in the device-independent context, thereby allowing us to {\em regularize} these raw data and obtain a direct estimation of the properties of interest  through the respective theoretical techniques.

Although there have been  attempts to perform regularizations for device-independent property estimations~\cite{Bancal2014,Schwarz2016} and for the quantification of nonlocality~\cite{Bernhard2014},  these proposals turn out to suffer from the  drawback of generating estimates that can be either {\em nonphysical} or {\em nonunique}. In contrast, our methods are provably free from such problems. Armed with these point estimates of the underlying distribution, a device-independent estimation of the property of interest then follows naturally by applying the algorithmic tools mentioned above (see Fig.~\ref{Fig:DI-estimation}).

\begin{figure}[h!]
	\scalebox{1}{\includegraphics{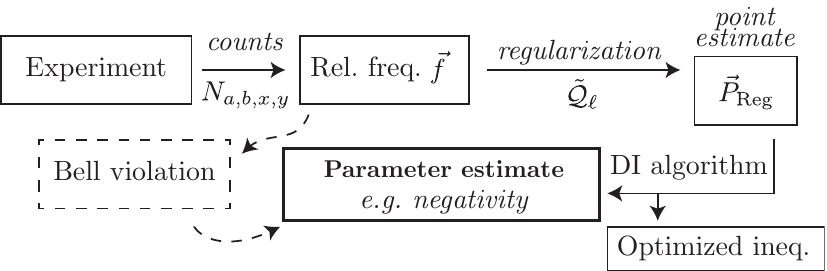}}
	\caption{\label{Fig:DI-estimation} Device-independent (DI) parameter estimation through the regularization of experimentally determined relative frequency $\vecf$. Here, $\AQ_\ell$ denotes an outer approximation of the quantum set $\Q$. The dashed lines show an alternative path to device-independent estimation using the strength of the observed Bell violation. Here, the quality of estimates relies crucially on the choice of Bell inequalities (see Fig.~\ref{Fig:Negativity:Mean} for an example). In contrast, we obtain, as a byproduct in our approach, an optimized device-independent witness (a Bell-like inequality) for the corresponding parameter estimation. See Appendix~\ref{App:Neg} for the connection between such a witness and an ordinary Bell inequality as well as the discussion at page~\ref{Sec:Discussion} for subtleties involved in using the Bell violation of $\vecf$ for device-independent parameter estimations.}
\end{figure}

\section{Preliminaries}
The starting point of device-independent estimations is a Bell experiment. Consider the simplest Bell scenario where Alice and Bob each randomly performs two possible measurements (labeled, respectively, by $x, y\in\{0,1\}$), and where each measurement gives binary outcomes (labeled, respectively, by $a, b\in\{0,1\}$). Generalization of our discussion to other finite Bell scenarios is obvious from the context. The correlations between their measurement outcomes can be summarized by a vector of joint conditional probability distributions $\vecP=\{P(a,b|x,y)\}_{a,b,x,y}\in\mathbb{R}^{16}$.

Denote by $\rho$ the state shared by Alice and Bob, and by $M^\text{A}_{a | x}$  ($M^\text{B}_{b | y}$) the positive-operator-valued-measure elements associated with their measurements. Born's rule dictates that for all $a,b,x,y$, the conditional probability distributions read as $P(a,b|x,y) \stackrel{\Q}{=} \tr \left( \rho\, M^\text{A}_{a | x} \otimes M^\text{B}_{b | y} \right)$
where the positivity and the normalization of probabilities demand that $M^\text{A}_{a | x},M^\text{B}_{b | y}\succeq0$ (matrix positivity) and $\sum_a M^\text{A}_{a | x}=\one_A, \sum_b M^\text{B}_{b | y}=\one_B$ with $\one_A,\one_B$ being identity operators. Throughout, we use $\Q$ to denote the set of quantum distributions, i.e., the collection of $\vecP$ that follows from Born's rule.

Importantly, quantum distributions satisfy the nonsignaling conditions~\cite{Popescu1994,Barrett2005}, i.e., their marginal distributions are independent of the measurement choice of the distant party:
\begin{equation}
\begin{split}\label{Eq:NS}
	\!\!\!\!\!\!P(a|x,y)\equiv\sum_b P(a,b|x,y) &= P (a | x, y'),\,\,\, \forall\, a,x, y, y',\\
	\!\!\!\!\!\!P(b|x,y)\equiv\sum_a P(a,b|x,y) &= P (b | x', y),\,\,\, \forall\,b, x, x', y.
\end{split}	
\end{equation} 
In an experiment, the underlying quantum distribution $P(a,b|x,y)$ is often estimated by computing the relative frequency, i.e.,  $P(a,b|x,y)\approx f(a,b|x,y)=\tfrac{N_{a,b,x,y}}{N_{x,y}}$ where $N_{a,b,x,y}$ is the number of coincidences registered for the combination of outcomes and settings $(a,b,x,y)$ while $N_{x,y} = \sum_{a b} N_{a,b,x,y}$ is the total number of trials pertaining to the measurement choice $(x,y)$. 

Of course, in the asymptotic limit of a large number of trials, i.e., when $\Nt=\min_{x,y} N_{x,y} \rightarrow \infty$, the difference between the true distribution $\vecP$ and the relative frequency $\vecf=\{f(a,b|x,y)\}_{a,b,x,y}$ vanishes. In practice, as $\max_{x,y} N_{x,y}$ is necessarily finite, not only is this difference nonzero but $\vecf$ typically also violates the weaker requirement of the nonsignaling conditions [see Eq.~\eqref{Eq:NS}]. As mentioned above, this discrepancy between theory and practice immediately renders many of the tools developed for device-independent characterizations inapplicable. 

\section{Regularization methods (estimators)}
To overcome this mismatch, one may project the observed frequency $\vecf$ onto an affine subspace $\NSspace$ of $\mathbb{R}^{16}$ which contains only $\vecP$'s that satisfy Eq.~\eqref{Eq:NS}. For example, if one  demands that this projection (via the corresponding projector $\Pi$) commutes~\cite{Renou2017} with {\em all} possible permutations of the labels for parties, settings, and outcomes (e.g., $a=0\leftrightarrow a=1$), then it happens to be equivalent to finding the {\em unique} minimizer of the least-square-error problem: $\vecPregM{$\Pi$}(\vecf)=\argmin_{\vecP\in\NSspace} ||\vecf-\vecP||_2$, where $\|\mathbf{.}\|_p$ denotes the $p$-norm; the regularization invoked in~\cite{Bancal2014} is precisely an application of such a projection (see Appendix~\ref{App:Projection}). 

Albeit intuitive and straightforward, such a projection suffers from the serious drawback that it may give ``negative probabilities" (see Appendix~\ref{App:Unphysical} for an explicit example). Indeed, the possibility of giving rise to an {\em  unphysical} estimate is a problem that such a projection shares with the linear inversion technique employed in standard quantum state tomography (see, e.g.,~\cite{Schwemmer2015}). Moreover, even when $\vecPregM{$\Pi$}(\vecf)$ represents a legitimate probability vector, it may well be outside the quantum set $\Q$. To overcome these issues, one is naturally led  to the \emph{least-square} (LS) estimator in the device-independent context, i.e., $\vecPnqa=\argmin_{\vecP\in\Q} ||\vecf-\vecP||_2$.

We thus see that various estimators  (see also~\cite{Schwarz2016,Bernhard2014} and Appendix~\ref{App:Other})  map $\vecf$ to a regularized distribution $\vecPreg(\vecf)$ that is non-negative and normalized and which satisfies the nonsignaling conditions. However, for a regularization method to be relevant for subsequent property estimation, it is also convenient that $\vecPreg(\vecf)$ is {\em uniquely} determined by $\vecf$. In particular, for a given $\vecf$, a nonunique estimator may give rise to  $\vecPreg(\vecf)$'s with drastically different properties, e.g., some being Bell-inequality violating (and therefore implying some nontrivial features of the underlying system) and some not [which renders that particular $\vecPreg(\vecf)$ useless for device-independent property estimation]. An \emph{ambiguity} then arises: which of these estimates should we rely on for subsequent property estimation? A possibility would be to consider the {\em worst case} over all such estimates, but this clearly complicates the property estimation as one would now need to consider an entire solution set $\{\vecPreg(\vecf)\}$ (the characterization of which is generally nontrivial). This makes evident the inconveniences of 1-norm estimators, and more generally nonunique estimators in the present context. The regularization procedure previously considered in~\cite{Schwarz2016,Bernhard2014}---both being 1-norm estimators---precisely suffers from this nonuniqueness drawback.

In this regard, note that a regularized distribution $\vecPreg(\vecf)$ obtained from minimizing a strictly convex function $g$ over a convex set (such as $\Q$) is provably unique, and is determined by $\vecf$ and $g$ (see, e.g., Theorem 8.3 of~\cite{Beck2018}). Using this observation, we show in Appendix~\ref{App:Uniqueness} that the aforementioned LS estimator is unique. Similarly, the device-independent analog of the ML estimator~\cite{Banaszek1999} is provably unique (see Appendix~\ref{App:Uniqueness} for a proof). To this end, consider the Kullback-Leibler (KL) divergence~\cite{vanDam2005,Acin2005,Kullback1951} (i.e., the relative entropy~{\cite{Cover:Book}}) from some $\vecP\in\Q$ to $\vecf$: 
\begin{equation}\label{Eq:KL}
	\DKL\left(\vecf||\vecP\right)=\sum_{a, b, x, y} f(x,y) f(a,b|x,y) \log_2 \left[ \frac{f(a,b|x,y)}{P(a,b|x,y)} \right],
\end{equation}
where $f(x, y)$ is the relative frequency of choosing the measurement settings labeled by $(x,y)$. 

The quantity $\DKL\left(\vecf||\vecP\right)$ can be seen as a measure of ``statistical closeness"~\cite{vanDam2005,Acin2005} between $\vecf$ and $\vecP$. Indeed, its minimization  over $\vecP\in\Q$ is equivalent~\cite{Cover:Book,vanDam2005} to maximizing the likelihood of producing the observed frequency by $\vecP\in\Q$ (we provide in Appendix~\ref{App:MLE} a proof adapted to the present context). The unique minimizer of $\DKL\left(\vecf||\vecP\right)$ over $\vecP\in\Q$, i.e., $\vecPkl=\argmin_{\vecP\in{\Q}} \DKL\left(\vecf||\vecP\right)$, therefore serves as the equivalent of the ML estimator in the device-independent context. Hereafter, we focus predominantly on this operationally well-motivated estimator. For further details of the LS estimator and some other plausible regularization methods, see, respectively, Appendix~\ref{App:NQA} and Appendix~\ref{App:Other}. 

As it stands, since there is no known exact characterization of $\Q$ using only finite resources, the ML estimator cannot be computed exactly.  Nonetheless, via a converging hierarchy of semidefinite programs (SDPs)~\cite{Navascues2007,Navascues2008a,Doherty2008,Moroder2013}, one can in principle obtain an arbitrary good approximation to $\Q$. To fix ideas, we hereafter focus on employing the  hierarchy $\AQ_\ell$ of approximations to $\Q$ discussed in~\cite{Moroder2013} and~\cite{Vallins2017}. The lowest level of this hierarchy $\AQ_1$ gives a decent outer approximation of $\Q$ known as the almost-quantum set~\cite{Navascues2015}. In general, $\AQ_\ell\subseteq\AQ_{\ell-1}$ for all $\ell\ge2$ and $\lim_{\ell\to\infty}\AQ_\ell\to\Q$. For any fixed $\ell$, although the nonlinear optimization problem $\argmin_{\vecP\in{\Q}} \DKL\left(\vecf||\vecP\right)$ does not appear to be a semidefinite program, we show in Appendix~\ref{App:KL-SCSetc} that it belongs to a more general class of convex optimization problems~\cite{Boyd2004Book} --- an exponential conic program. A minimization of the KL divergence with NPA constraints is thus also efficiently solvable on a computer with a numerical precision of $10^{-6}$ or better.

\section{Notable properties of point estimates}
The uniqueness of $\vecPkl$ and the nonnegativity of the KL divergence ensure that our estimators are {\em consistent}~\cite{Shao2003Book}, in the sense that they provide an estimate that converges to the true distribution $\vecPQ$ in the asymptotic limit of $\Nt\to\infty$. This can be seen by noting that in the asymptotic limit,  $\vecf\to\vecPQ$, the nonnegativity of the KL divergence then implies that the unique minimizer of $\DKL\left(\vecf||\vecP\right)$ over $\vecP\in\Q$ is necessarily given by $\vecP=\vecPQ$. [Likewise, the LS estimator $\vecPnqa$ is provably consistent.]

In practice, one would evidently be more interested  in how these methods fare for finite values of $\Nt$. To gain insights into this, we carry out extensive numerical simulations by (i) picking some ideal $\vecPQ$, (ii) numerically simulating the outcomes of a Bell experiment according to $\vecPQ$ and computing the relative frequency $\vecf$, (iii) computing the point estimate $\vecPreg(\vecf)$ and calculating various quantities of interest, and (iv) repeating steps i-iii $10^4$ times for $\Nt=10^2,10^3,\ldots,10^{10}$ for all $x,y$.\label{Procedure} For simplicity, we take $N_{x,y}=\Nt$, i.e., a constant independent of $x,y$ [this amounts to setting $f(x, y)$ as a constant in Eq.~\eqref{Eq:KL}]. 

Our numerical results in Appendix~\ref{App:Convergence} suggest that {\em in general}, the difference between $\vecPreg(\vecf)$ and the ideal distribution $\vecPQ$, quantified, e.g., via $||\vecPreg(\vecf)-\vecPQ||_1$ (or other $p$ norms), diminishes, as with $||\vecf-\vecPQ||_1$, at a rate proportional to $\tfrac{1}{\sqrt{\Nt}}$. Similar convergence is also observed for $\DKL\left(\vecPkl||\vecPQ\right)$. Moreover, although we have only employed an outer approximation to $\Q$ in the regularization step iii, as our example below illustrates, the regularized distribution $\vecPreg(\vecf)$ can already be used to perform reasonable device-independent property estimations. 

\section{Application to device-independent estimations}
\label{Sec:DI}
As a concrete example of such property estimations via regularization (see Fig.~\ref{Fig:DI-estimation}), consider $\vecPQ=\vecPQ^{\tiny \tau_{1.25}}$, a quantum distribution considered in a  recent Bell test~\cite{Christensen2015} (see Appendix~\ref{App:PQ} for details). Device-independent estimations of the underlying negativity~\cite{Vidal2002} (a well-known entanglement measure) based on {\em ideal} quantum distributions are known to be possible~\cite{Moroder2013}. Here, we illustrate how such an estimation can be realized for {\em finite} data through the regularization of $\vecf$. To facilitate comparison, we plot in Fig.~\ref{Fig:Negativity:Mean} the {\em average} negativity N$(\rho)$ of the underlying state $\rho$ estimated from the regularized distribution $\vecPreg(\vecf)$ (via the SDP described in~\cite{Moroder2013}) against that deducible from the amount of Clauser-Horne-Shimony-Holt~\cite{CHSH} (CHSH) Bell-inequality violation $\S_\CHSH$~\cite{Moroder2013}: N$(\rho)\ge \tfrac{\S_\CHSH-2}{4\sqrt{2}-4}$ for $\S_\CHSH\in[2,2\sqrt{2}]$.

\begin{figure}[h!]
	\scalebox{0.25}{\includegraphics{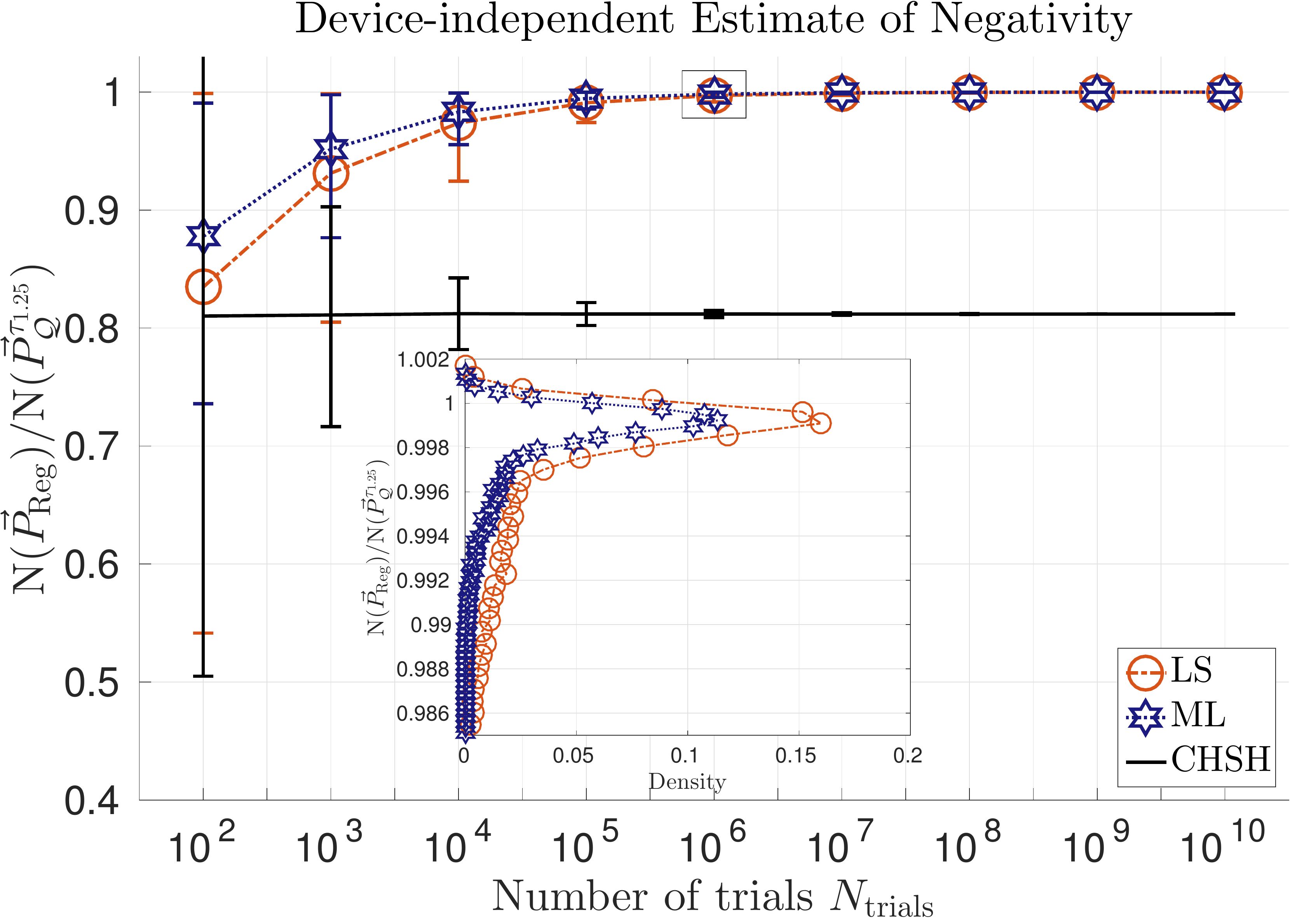}}	
	\caption{\label{Fig:Negativity:Mean} (Color) Mean value of the normalized negativity estimated from the regularized distributions (average over $10^4$ runs) as a function of $\Nt$ based on the relative frequencies $\vecf$ generated from $\vecPQ^{\tau_{1.25}}$ (see Appendix~\ref{App:PQ}), which has N$(\vecPQ^{\tau_{1.25}})\approx 0.38$.  In estimating the negativity, we feed the distributions regularized to $\AQ_2$ into the second level SDP of [19]. The lower and upper limit of each error bar mark, respectively, the 10\% and 90\% window of the spread of the negativity value. The inset shows the corresponding histograms for $\Nt=10^6$. On average, the ML estimator performs considerably better than the LS estimator.}
\end{figure}

A few features of this comparison are worth noting. First, since the negativity estimated directly from the CHSH Bell-inequality violation $\S_\CHSH$ of $\vecf$ depends linearly on this violation, the mean value of the negativity estimated hardly depends on $\Nt$, and is suboptimal. In fact, even in the asymptotic limit, a negativity estimation based on $\vecPQ^{\tau_{1.25}}$ and $\S_\CHSH$ is suboptimal. In contrast, the mean value of the negativity estimated from the regularized distribution $\vecPreg(\vecf)$ rapidly converges to the true value as $\Nt$ increases; already at $\Nt=10^4$, this mean value only differs from the true value by a few percents. Second, note that our negativity estimations based on $\vecPreg(\vecf)$ clearly systematically {\em underestimate}  the amount of negativity present, which is in strong contrast with the results presented in~\cite{Schwemmer2015} for the {\em non-device-independent} scenario using both least-square and maximum-likelihood estimators. (For further examples of underestimation using $\vecf$ sampled from other quantum distributions, see Appendix~\ref{App:Biased}.) 

Of course, instead of the CHSH Bell inequality, one could hope to improve the negativity estimation by considering a Bell-inequality the quantum violation of which is optimized for the negativity estimation of $\vecPQ^{\tau_{1.25}}$ [see Eq.~\eqref{Ineq:Itau} and Appendix~\ref{App:Neg} for details]. Such an optimized device-independent witness may be obtained, e.g., by feeding the DI algorithm with some regularized distribution sampled from $\vecPQ^{\tau_{1.25}}$, as indicated in Fig.~\ref{Fig:DI-estimation}. In practice, however, the relative frequencies sampled from $\vecPQ^{\tau_{1.25}}$ turn out to give---independent of $\Nt$---about half the time, a Bell violation more than that allowed by quantum theory, thereby rendering negativity estimation from this Bell violation impossible in all these cases.

\section{Discussion}
\label{Sec:Discussion}
The device-independent state estimation problem is the core of all state estimation problems from finite data, as it addresses the generic problem of matching empirical data (subjected to statistical fluctuations) with the ideal, theoretical description given by Born's rule. The recent demonstration of loophole-free Bell tests~\cite{hensen_loophole-free_2015, shalm_strong_2015, giustina_significant-loophole-free_2015,rosenfeld_event-ready_2017} has made it clear that the development of reliable techniques for the device-independent estimation of underlying  properties from {\em finite data} not only is of fundamental interest but also would play an indispensable role in the next generation of quantum information protocols.  In fact, although our focus is on a fully device-independent setting, the insights obtained thereof are also relevant in analogous problems in a partially device-independent scenario, such as those incurred in a quantum steering experiment~\cite{Cavalcanti:2015aa}.

In this paper, we have provided the device-independent analog of the maximum-likelihood and the least-square estimators and shown that they are physical, computationally tractable, and unique.  These features render them ideal for bridging the device-independent tools developed for ideal quantum distributions and the experimentally obtained raw data. Generalizing the arguments in~\cite{Schwemmer2015}, however, it can be shown that {\em all} quantum estimators are necessarily biased, as summarized in the following proposition (see Appendix~\ref{App:Biased} for a proof and the corresponding definition of {\em being strictly nonorthogonal}).
\begin{proposition}\label{Prop:Biased}
  Let $\C$ be a closed convex subset of the nonsignaling polytope $\NSset$ (such as $\Q$) with two strictly nonorthogonal extreme points. Any point estimator constrained to give $\vecPreg(\vecf)\in\C$ is necessarily biased.
\end{proposition}

As stressed above, one of the goals of performing a device-independent state estimation is to obtain therefrom various device-independent parameter (property) estimates (see~Fig.~\ref{Fig:DI-estimation}). Given Proposition~\ref{Prop:Biased}, one may expect that any such parameter estimates are also biased. Indeed, for data sampled from a distribution that is not Bell-inequality violating, but which is sufficiently {\em near} to the boundary of the local polytope, an {\em overestimation} of the corresponding negativity is to be expected.\footnote{Evidently, due to statistical fluctuations, some of the $\vecPreg(\vecf)$ would lie inside the local polytope, giving a zero negativity value while some other $\vecPreg(\vecf)$ would violate a Bell inequality, giving a strictly positive negativity value. Their average is thus positive, thereby resulting in an overestimation.} On the other hand, our numerical results show that  for data sampled from some extremal quantum distribution, we observe, instead, an {\em underestimation} of the underlying negativity (see Fig.~\ref{Fig:Negativity:Mean}) and/or Bell violation (see Fig.~\ref{Fig:Bias}). Admittedly, for such point estimation to be useful, a more thorough investigation is needed in order to determine when such device-independent estimations are trustworthy (in the sense of {\em not} leading to an overestimation).

\label{Pg:subtlety}As one can perform analogous property estimation directly from the observed Bell-inequality violation, our approach of {\em estimation by regularization} may seem redundant at first glance. However, as we illustrate in the negativity example, the quality of an estimate obtained from Bell-inequality violation depends heavily on the choice of the inequality and it is {\em a priori} not always obvious which Bell inequality (the violation of which is used as a device-independent witness) will provide an optimal estimate. In contrast, our approach yields an {\em optimized} Bell-like inequality for witnessing the desired property as a byproduct. Moreover, due to the signaling nature of the relative frequencies, even physically equivalent~\cite{Collins2004} Bell inequalities may give rise to different estimates~\cite{Renou2017}, thereby resulting in further ambiguities.

We now briefly comment on some other possibilities for future research. Implicit in our discussion is the assumption that the experimental trials are independent and identically distributed. While this is an often adopted assumption in the non-device-independent setup, its justification in a device-independent setting is far from trivial. A natural line of research thus consists of relaxing this assumption while maintaining the possibility to perform a reliable state estimation. 

In close connection to this is the problem of establishing a confidence region: how does one generalize the tools presented here to construct a region of estimates in accordance to, say, some pre-defined likelihood~\cite{arXiv:1202.5270, Christandl2012, Sugiyama2013, Shang2013, Faist2016}? For some device-independent tasks, specific techniques~\cite{Zhang2011,Zhang2013,Pironio2013,arXiv:1611.00352:brief,arXiv:1702.05178:brief,Dupuis:1607.01796:brief,Rotem2017} for dealing with finite statistics (possibly with the inclusion of confidence regions) have been developed, but general techniques for establishing confidence regions associated with generic quantum properties are still lacking. To appreciate the importance of constructing these confidence regions, recall from our example that the Bell violation given by the observed relative frequency---due to statistical fluctuations---may give rise to a value beyond that allowed by quantum theory, thus rendering this Bell violation useless for the estimate of a quantum parameter. To this end, we remark that the work of~\cite{Wills:Unpublished} suggests that the hypothesis-testing technique of~\cite{Zhang2011}, together with the numerical technique developed here, can indeed be used to provide a confidence region for general device-independent estimates. Addressing these questions, however, clearly goes beyond the scope of the present paper and is something that we plan to take up in the sequel to this paper.

\begin{acknowledgments}
We are grateful to B\"anz Bessire, Daniel Cavalcanti, Flavien Hirsch, Sacha Schwarz and Paul Skrzypczyk for useful discussions, and to Boris Bourdoncle, Jedrzej Kaniewski, Lukas Knips, as well as a few anonymous referees for useful comments on an earlier version of this paper. This work is supported by the Ministry of Science and Technology, Taiwan (Grant No. 104-2112-M-006-021-MY3), the National Center for Theoretical Sciences of Taiwan, the Swiss National Science Foundations [through the NCCR QSIT as well as Grants No. PP00P2-150579 and No. P2GEP2$\_$162060 (Early Postdoc.Mobility)],  the Ontario Research Fund, and the Natural Sciences and Engineering Research Council of Canada.
\end{acknowledgments}

\appendix

\section{Notations and definitions}
\label{App:GeneralP}

Throughout, we  label the measurement settings (inputs) by $x,y,z,\ldots$ and the corresponding measurement outcomes by $a,b,c,\ldots$ where each of these labels are elements from some finite sets.  The correlations between their measurement outcomes are succinctly summarized by the vector of joint conditional distributions $\vecP:=\{P(a,b,c,\ldots|x,y,z,\ldots)\}$. For simplicity, we will focus our discussion on the bipartite cases.

We denote by $\Pset$ the set of all legitimate probability distributions, i.e., those that obey:
\begin{subequations}
\label{Eq:App:Signaling}
\begin{align}
 \label{Eq:ConeNorm}
 \sum_{a,b} P(a,b|x,y)  &= 1\quad \forall\quad x,y,    \\
\label{Eq:App:Signaling:NonNeg}
P(a,b|x,y)                              & \geqslant 0\quad \forall\quad x,y,a,b.
\end{align}
\end{subequations}
Moreover, we denote by $\N$ the subset of $\Pset$ which satisfies {\em also} the {\em nonsignaling conditions}~\cite{Popescu1994,Barrett2005}:
\begin{equation}
\begin{split}\label{App:Eq:NS}
	\!\!\!\!\!\!P(a|x,y)\equiv\sum_b P(a,b|x,y) &= P (a | x, y'),\, \forall\, a,x, y, y',\\
	\!\!\!\!\!\!P(b|x,y)\equiv\sum_a P(a,b|x,y) &= P (b | x', y),\, \forall\,b, x, x', y.
\end{split}	
\end{equation} 
It is worth noting that both $\Pset$ and $\NSset$~\cite{Barrett2005} are convex polytopes, i.e., convex sets having a finite number of extreme points, and are conventionally referred to, respectively, as the signaling and the nonsignaling polytope. Here, we describe both $\Pset$ and $\NSset$ in their H-representation using linear equations and inequalities.

\begin{definition}
\label{Def:NSaffine}
The nonsignaling affine space $\NSspace \supset \NSset $ is the smallest-dimensional affine space containing the set $\NSset$ and is given by the distributions satisfying Eqs.~\eqref{Eq:ConeNorm} and~\eqref{Eq:NS}. 
\end{definition}

\section{Further details about the projection method}
\label{App:Projection}

Here, we provide some further details about the projection mentioned in the main text.

\subsection{Equivalent definitions}
\label{App:Projection:Def}
The projection method can be defined in three equivalent ways. 

\begin{enumerate}

\item It is the minimizer $\vecPpi$ of the following optimization problem: 
  \begin{equation}
    \label{Eq:App:Projection:AsMinimization}
    \min_{\vecP\in \NSspace} \left\| \vecP - \vecf\, \right\|_2 = \left\| \vecPpi - \vecf\, \right\|_2
  \end{equation}
  
\item It is the nonsignaling part $\vecf_{\NSspace} \in \NSspace$, i.e., the first component of the decomposition
  \begin{equation}
    \label{Eq:App:Projection:AsOrthogonal}
    \vecf = \vec{f}_{\NSspace} + \vec{f}_\text{SI},
  \end{equation} 
  where $\vec{f}_\text{SI}$ is the signaling component of $\vecf$ and is orthogonal to all vectors in the affine subspace $\NSspace$.
  
\item It is the result
  \begin{equation}
    \label{Eq:App:Projection:AsProjection}
    \vecPpi = \Pi \vecf
  \end{equation}
  of the projection onto $\NSspace$ by the linear operator $\Pi$, where $\Pi$ is uniquely determined by the set of commutation relations $\Pi M = M \Pi$ with $M$ being {\em any} permutation matrix corresponding to relabeling~\cite{Collins2004} of outputs, inputs and parties.
  
\end{enumerate}
For the bipartite Bell scenario with binary inputs and outputs, their equivalence can be shown using the decomposition given in~\cite[Prop. in Sec. 3]{Renou2017}. For more general Bell scenarios, a proof of their equivalence analogously follows from group representation theory, and will be made available in~\cite{Rosset:Unpublished}.

\subsection{Explicit form of the projection matrix in the simplest Bell scenario}

Using the notation of~\cite{Renou2017}, the projection operator $\Pi$ in the bipartite Bell scenario with binary inputs and outputs admits the explicit form
\begin{align}\label{Eq:Pi}
\Pi_{abxy,a'b'x'y'} = \mathbbm{1}_{16} - \frac{1}{16} \sum_{(i,j,k,l)\in \mathcal{I}} i^{a+a'} j^{b+b'} k^{x+x'} l^{y+y'},
\end{align}
where $\mathbbm{1}_{16}$ is the 16$\times$16 identity matrix, the sum is carried out over the four quadruplets $\mathcal{I} = \{(+1,-1,-1,\pm 1),(-1,+1,\pm 1, -1)\}$, the rows of $\Pi$ are indexed by $(a,b,x,y)$, while its columns are indexed by $(a',b',x',y')$.

\subsection{An algorithm for performing the projection for the more general Bell scenarios}
\label{App:Algo}

In general, in (bipartite) Bell scenarios where parties have binary outputs, it is customary to write $A = (-1)^a$ and $B = (-1)^b$, and compute the expectation values (correlators~\cite{Bancal2012}):
\begin{equation}
  \label{Eq:correlator:oneparty}
  \left< A_x \right> = \sum_{A=\pm1} A ~ P(A|x) = \sum_{a=0,1} (-1)^a P(a|x)
\end{equation}
and similarly for $\left< B_y \right>$, while
\begin{equation}
  \label{Eq:correlator:twoparties}
  \left< A_x B_y \right> = \sum_{A, B} A ~ B ~ P(a,b|x,y) = \sum_{a,b} (-1)^{a+b} P(a,b|x,y).
\end{equation}

Together, the $\left<A_x\right>$, $\left<B_y\right>$ and $\left< A_x B_y \right>$ represent a parametrization of the nonsignaling subspace. In particular, the above relations can be inverted to give:
\begin{equation}\label{Eq:Cor2Prob}
\begin{split}
	P(a,b|x,y)=\frac{1}{4}\big[1&+(-1)^a\left<A_x\right>+(-1)^b\left<B_y\right>\\
					&+(-1)^{a+b}\left<A_xB_y\right>\big]
\end{split}
\end{equation}
However, when the distribution $P(a,b|x,y)$ is signaling, the marginals $P(a|x,y)$, $P(b|x,y)$ depend on both inputs $(x,y)$.
As shown in~\cite{Renou2017}, in the case of binary inputs, the (nonsignaling) correlators $\left<A_x\right>$ should be computed as $\left< A_x \right> = \sum_{a} (-1)^a \tilde{P}_\text{A}(a|x)$ where $\tilde{P}_\text{A}(a|x)$ is given by:
\begin{equation}
  \label{Eq:correlator:average}
  \tilde{P}_\text{A}(a|x) = \frac{1}{2} \sum_{b,y} P(a,b|x,y)
\end{equation}
i.e., averaged uniformly over $y=0,1$. This is the only choice that keeps the projection invariant under permutations of $y$.
We make the same choice to compute $\left< B_y \right>$.
Now, the correlators $\left<A_x\right>$, $\left<B_y\right>$ and $\left< A_x B_y \right>$ correspond to a unique distribution $P_\Pi(a,b|x,y)$ in the nonsignaling subspace, cf. Eq.~\eqref{Eq:Cor2Prob}, which is the result of the projection. This is the essence of the regularization method employed in~\cite{Bancal2014}.

The construction above extends to an arbitrary number of parties, inputs, and outputs, specified by the tuple $(n,m,k)$.
We give a sketch below of this generalization, which satisfies criteria \eqref{Eq:App:Projection:AsMinimization}--\eqref{Eq:App:Projection:AsProjection} (a proper proof will be discussed in a future work~\cite{Rosset:Unpublished}).

The generalization to additional parties and inputs is simple.
Write, for example:
\begin{equation}
  \label{Eq:averaged:marginal}
  \tilde{P}_\text{A}(a|x) = \frac{1}{m^{n-1}} \sum_{b,y,c,z...} P(a,b,c,...|x,y,z...),
\end{equation}
and the same for other marginal distributions $\tilde{P}_\text{...}$, by averaging uniformly over all inputs {\em not} fixed by the indices of the marginals $\tilde{P}_\text{...}$.
Then, use these $\tilde{P}_\text{...}$ to compute the correlators according to the straightforward multipartite generalizations of Eqs.~\eqref{Eq:correlator:oneparty} and \eqref{Eq:correlator:twoparties}.

To cater for scenarios with nonbinary outputs, the correlators have to be generalized.
We use the framework proposed in~\cite{Bancal2010}, adapting slightly the notation to the present paper:
\begin{equation}
  \label{Eq:correlator:general}
  \left< A^i_x \right > = \sum_a c_{ia} \tilde{P}(a|x), \quad 0 \le i \le k-2,
\end{equation}
where $c_{ia} = k ~ \delta_{ia} - 1$ and we omitted the coefficient for $i=k-1$ as it is linearly dependent on the others.
Note that $\left< A^0_x \right> = \left< A_x \right>$ in the case of binary outcomes ($k=2$).
Multipartite correlators are written, for example:
\begin{equation}
  \label{Eq:correlator:multipartite}
  \left< A^i_x B^j_y \right > = \sum_{a,b} c_{ia} c_{jb} \tilde{P}(a,b|x,y), \quad 0 \le i,j \le k-2.
\end{equation}

Starting from a signaling distribution $P(a,b,\ldots|x,y,\ldots)$, we compute the generalized correlators using the averaged marginals of Eq.~\eqref{Eq:averaged:marginal}.
Then we interpret the resulting correlators as a description of the nonsignaling distribution $P_\Pi(a,b,\ldots|x,y,\ldots)$. The whole process is a projection: all operations are linear, and the averaging [e.g., $P(a|x,y) \rightarrow \tilde{P}(a|x)$] is injective for nonsignaling distributions. Importantly, the projection defined by this algorithm commutes with relabelings and thus corresponds to the three equivalent definitions given in Appendix~\ref{App:Projection:Def}.

\subsection{An explicit example showing that the output of the projection method may be nonphysical}
\label{App:Unphysical}

Although easy to compute, the projection method may give rise to coefficients of $\vecPpi$ that are negative. Indeed, the output space of the projection method is the nonsignaling {\em affine space} $\NSspace$. We now give an explicit example to illustrate this fact. Consider some relative frequency $\vecf$ given in the compact matrix representation:

\begin{equation}\label{Eq:vecf:eg}
\!\!\!\!\vecf=\left[\begin{array}{c|c}
f(a,b|0,0) & f(a,b|0,1)  \\\hline
f(a,b|1,0) & f(a,b|1,1)  \\
\end{array}\right]
=\tfrac{1}{10}\left[
\begin{array}{cc|cc}
3 & 0 & 7 & 0 \\
1 & 6 & 1 & 2 \\\hline
5 & 1 & 1 & 6 \\
1 & 3 & 3 & 0 \\
\end{array}\right]
\end{equation}
where the entries in each block are arranged such that the value of $a$ ($b$) increases downward (rightward). By applying the projection matrix $\Pi$ given in Eq.~\eqref{Eq:Pi} to this signaling distribution, one obtains:
\begin{equation*}\label{Eq:Pi-output}
	\vecPpi\!=\!\left[\begin{array}{c|c}
	P(a,b|0,0) & P(a,b|0,1)  \\\hline
	P(a,b|1,0) & P(a,b|1,1)  \\
	\end{array}\right]\!=
	\!\tfrac{1}{40}\!\!\left[
	\begin{array}{cc|cc}
	18 & 2 & 20 & 0 \\
	2 & 18 & 4 & 16 \\\hline
	19 & 7 & 7 & 19 \\
	1 & 13 & 17 & -3 \\
\end{array}\right],
\end{equation*}
which is easily seen to satisfy the nonsignaling conditions of Eq.~\eqref{Eq:NS}. However, this $\vecPpi$ is evidently nonphysical as its entry for $x=y=a=b=1$ is negative.

\section{Further details about the device-independent least-square method}
\label{App:NQA}

Here, we present the details of the device-independent analog of the least-square tomography method. Formally, the method amounts to finding the {\em unique} minimizer of the least-square problem: 
\begin{equation}
 \label{Eq:NQA2Method}
  \vecPnqa=\argmin_{\vecP\in\Q} ||\vecf-\vecP||_2.
\end{equation}

\subsection{Equivalence to performing a projection and minimization of the 2-norm distance from $\vecPpi$ to $\Q$}
\label{App:Decomposition}

We now prove that the above optimization is equivalent to first performing the projection method, followed by performing a minimization of the 2-norm distance from $\vecPpi$ to $\Q$. For convenience, we shall prove this equivalence, instead, for any converging superset relaxation~\cite{Navascues2008a,Doherty2008,Moroder2013} $\AQ_\ell$ of the quantum set $\Q$. The desired equivalence then follows from the fact that $\lim_{\ell\to\infty} \AQ_\ell =\Q$.

\begin{lemma}\label{NQA-Proj-NQA}
Given the relative frequency $\vecf$, the least-square estimator (when $\AQ_\ell$ is used to approximate $\Q$) satisfies $\vecPnqa=\argmin_{\vecP\in\AQ_\ell} \|\vecP-\vecPpi\|_2$.
\end{lemma}

\begin{proof}
 First, note from Appendix~\ref{App:Projection:Def} that any given relative frequency $\vecf$ can be decomposed as $\vecf = \vecf_{\overline\NSset}+\vecf_{\rm SI}$, where $\vecf_{\overline\NSset}\in\overline{\NSset}$ is orthogonal to $\vecf_{\rm SI}$. Similarly, for any $\vec{P} \in \AQ_\ell \subset \NSset \subset \overline\NSset$, we must have $(\vec{P}-\vecf_{\overline\NSset}) \in \overline\NSset$, which is orthogonal to $\vecf_{\rm SI}$. It then follows from the definition of the least-square method that
 \begin{align*}
   \vecPnqa &= \argmin_{\vec{P}\in\AQ_\ell} \|\vec{P}-\vecf\|_2  = \argmin_{\vec{P}\in\AQ_\ell} \|(\vec{P}-\vecf_{\overline\NSset})-\vecf_{\rm SI}\|_2\\
            & = \argmin_{\vec{P}\in\AQ_\ell} \|\vec{P}-\vecf_{\overline\NSset}\|_2 = \argmin_{\vec{P}\in\AQ_\ell} \|\vec{P}-\vecPpi\|_2        
 \end{align*}
where the second to last equality follows from the orthogonality of $\vec{P}-\vecf_{\overline\NSset}$ and $\vecf_{\rm SI}$ and the fact that $\|\vecf_{\rm SI}\|^2$ is a {\em constant} in the 2-norm minimization, while the last equality follows from the equivalence between the second and the third definition of the projection method.
\end{proof}
Note that although $\vecPpi$ is not necessarily in $\AQ_\ell$, if it happens that $\vecPpi\in\AQ_\ell$ then the equivalent regularization shown in Appendix~\ref{App:NQA} implies that $\vecPnqa=\vecPpi$.

\subsection{Formulation as a semidefinite program}
\label{App:NQA2-SDP}

When the quantum set $\Q$ is approximated by a superset relaxation $\AQ_\ell$ that admits a semidefinite programming~\cite{Boyd2004Book} characterization, the optimization problem of Eq.~\eqref{Eq:NQA2Method} can also be solved as an SDP. To this end, note that the SDP characterization of $\AQ_\ell$ is achieved in terms of some moment matrix $\chi$ that contains all the entries of $\vecP$ as some of its matrix elements. 

Using the characterization of positive semidefinite matrices via their Schur complements (see, e.g., Theorem 7.7 of~\cite{Horn1985Book}), we can then reformulate Eq.~\eqref{Eq:NQA2Method} (with $\AQ_\ell$ approximating $\Q$) as: 
\begin{equation}\label{Eq:NQA2b}
\begin{split}
	&\argmin_{\vecP\in\AQ_\ell}\ \ \ \ \qquad s \\
	&\text{s.t.}\quad 
	\begin{pmatrix} 
	s \mathbbm{1} & \vecf-\vecP \\ 
	\vecf\top-\vecP\top & s
	\end{pmatrix} \succeq 0,
\end{split}
\end{equation}
where $\mathbbm{1}$ is the identity matrix having the same dimension as the column vector $\vecf$, and $\vecf\top$ is the transpose of $\vecf$. Evidently, we see from Eq.~\eqref{Eq:NQA2b} that the equivalent optimization problem of Eq.~\eqref{Eq:NQA2Method} now involves only an objection function and matrix inequality constraints  that are {\em linear} in all their optimization variables: $s$, $\vecP$ and some other entries of $\chi$ (that cannot be estimated from experimental data). Thus, as claimed, the minimization of the 2-norm of $\vecf-\vecP$ over $\vecP\in\AQ_\ell$ can indeed be cast as an SDP.

\section{Further details about the Kullback-Leibler divergence and the corresponding regularization method}
\label{App:KL}

The Kullback-Leibler (KL) divergence:
\begin{equation}
\label{Eq:App:KL} 
	\DKL\left(\vec{v}||\vec{w}\right)=\sum_{i} v_i \log_2 \left( v_i / w_i \right),
\end{equation}
is conventionally defined for unconditional probability distributions $v_i$/ $w_i$ [such as $P(a,b,x,y)$ and $f(a,b,x,y)$]. In the explicit examples studied, we fixed $P(x,y) = f(x, y) = 1/|\X \times \Y|$. This allows us to keep our definition of the ML method valid when $f(x,y)$, sampled from $P(x,y)=$ constant, has itself statistical fluctuations, as reflected in our definition of the KL divergence given in Eq.~(2) in the main text. 

Note that while the KL divergence is a statistical distance, it is {\em not} a metric as it is asymmetric and violates the triangle inequality. To appreciate the relevance of this asymmetry, see~\cite{vanDam2005}.

\subsection{Connection to maximum likelihood}
\label{App:MLE}

The equivalence between the minimization of the KL divergence over some $\vecP\in\C$ (for some set $\C$) to $\vecf$ and the maximization of the likelihood of generating $\vecf$ from $\vecP$ can be seen as follows:
\begin{equation}\label{Eq:KL:likelihood}
\begin{split}
	&\min_{\vecP\in\C} \DKL\left(\vecf||\vecP\right)\\
	=&\min_{\vecP\in\C}\sum_{a, b, x, y} f(x,y) f(a,b|x,y) \log_2 \left[ \frac{f(a,b|x,y)}{P(a,b|x,y)} \right],\\
	=&\kappa +\min_{\vecP\in\C}-\tfrac{1}{N}\sum_{a, b, x, y} N_{a,b,x,y}  \log_2 P (a, b |x, y),\\
	=&\kappa -\tfrac{1}{N}\max_{\vecP\in\C} \sum_{a, b, x, y}   \log_2 P (a, b |x, y)^{N_{a,b,x,y}},\\
	=&\kappa -\tfrac{1}{N}\max_{\vecP\in\C}  \log_2 \prod_{a, b, x, y} P (a, b |x, y)^{N_{a,b,x,y}},\\	
\end{split}
\end{equation}
where $\kappa:=\sum_{a, b, x, y} f(x,y)f(a, b| x, y)  \log_2 f(a, b |  x, y)$ is a constant of the optimization, $N:=\sum_{x,y,a,b} N_{a,b,x,y}$, and in the second equality we have used the definition of the relative frequency $\vecf$ and the fact that $f(x, y)=\tfrac{N_{x,y}}{N}$. In the last line of Eq.~\eqref{Eq:KL:likelihood}, the argument of the maximization is the log likelihood of observing $N_{a,b,x,y}$ times the event labeled by $(x,y,a,b)$  with probability $P(a,b|x,y)$. Hence, we see that the minimization of the KL divergence is equivalent to maximizing the likelihood of generating $\vecf$ given $P(a,b|x,y)$.

\subsection{Formulation as a conic program}
\label{App:KL-SCSetc}

Here, we briefly explain how the device-independent ML method can be formulated and solved as a conic program (CP) with an exponential cone. Recall from the main text that for any given relative frequency $\vecf$, the ML estimator (with the quantum set approximated by $\AQ_\ell$) works by solving:
\begin{equation}\label{Eq:KL-method}
	\argmin_{\vecP\in \AQ_\ell}\quad \sum_{a, b, x, y} f(x,y) f(a,b|x,y) \log_2 \left[ \frac{f(a,b|x,y)}{P(a,b|x,y)} \right].
\end{equation}

A conic program takes the canonical form of:
\begin{align}
\label{Eq:ConicProgram}
  \min \quad & \vec{c}\cdot \vec{x}, \nonumber \\
  \text{s.t.} \quad & A\vec{x} = \vec{b}, \nonumber \\ 
  & \vec{x} \in K,
\end{align}
where $K$ is a convex cone, such as an exponential cone:
\begin{equation}
\label{Eq:ExponentialCone}
K_\text{exp}=\{(u,v,w) ~ | ~ v e^{u/v} \le w, v\ge 0\}.
\end{equation}

After discarding the constant term and folding the factors in $f(a,b,x,y) = f(x, y) f(a,b|x,y)$, the minimizer of Eq.~\eqref{Eq:KL-method} can be obtained by solving:
\begin{align}
  \vecPkl = \argmin \quad & \sum_{abxy} f(a,b,x,y) \log_2 \frac{1}{P(a,b|x,y)} \nonumber \\
  \text{s.t.} \quad \chi & \succeq 0, \nonumber \\
             \tr[ F_{abxy} \chi ] & = P(a,b|x,y) \quad \forall\quad a,b,x,y, \nonumber \\
  \tr[ G_k \chi ] & = 0, \quad k = 1,2,...
\end{align}
where $\chi$ is the moment matrix associated with $\AQ_\ell$, while $F_{abxy}$ and $G_k$ encode the equality constraints associated with the structure of this moment matrix. This problem has the conic form of Eq.~\eqref{Eq:ConicProgram}:
\begin{subequations}
\begin{align}
  \vecPkl = \argmax \quad & \sum_{abxy} f(a,b,x,y) u_{abxy} \\
  \text{s.t.} \quad \chi & \succeq 0, \label{Eq:ConicConstraintSDP}\\
                          e^{u_{abxy}} & \le P(a,b|x,y) \quad \forall\quad a,b,x,y, \label{Eq:ConicConstraintExponential} \\
  \tr[ F_{abxy} \chi ] & = P(a,b|x,y) \quad \forall\quad a,b,x,y, \\
  \tr[ G_k \chi ] & = 0, \quad k = 1,2,...
\end{align}
\end{subequations}
where the constraint~\eqref{Eq:ConicConstraintSDP} is that for a positive semidefinite  cone and the constraint~\eqref{Eq:ConicConstraintExponential} is that for copies of the exponential cone~\eqref{Eq:ExponentialCone} (with dummy variables $v_{abxy} = 1$ for all $a,b,x,y$). Thus, we see that the problem of Eq.~\eqref{Eq:KL-method} is indeed an exponential conic program.

\section{Details of numerical investigations}
\label{App:Numerics}

\subsection{Explicit form of the quantum distribution $\vecPQ$ considered}
\label{App:PQ}

Here, we provide the explicit form of the ideal quantum distributions $\vecPQ$ employed in our numerical studies and an explicit quantum strategy realizing each of these correlations. First, we consider $\vecPQ=\vecPQ^\text{\tiny CHSH}$ with entries given by $\tfrac{1}{4}+(-1)^{a+b+xy}\tfrac{\sqrt{2}}{8}$. $\vecPQ^\text{\tiny CHSH}$ is known to violates maximally the CHSH~\cite{CHSH} Bell inequality: 
\begin{equation}\label{Ineq:CHSH}
	\I_\text{\tiny CHSH} = \sum_{a,b,x,y=0}^1 (-1)^{a+b+xy} P(a,b|x,y) \stackrel{\text{\tiny $\L$}}{\le} 2,
\end{equation}
up to the limit of 2$\sqrt{2}$ allowed by quantum theory. $\vecPQ^\text{\tiny CHSH}$ can be realized by both parties locally measuring  $\cos\tfrac{3\pi}{8}\sigma_z+\sin\tfrac{3\pi}{8}\sigma_x$ and $\cos\tfrac{7\pi}{8}\sigma_z+\sin\tfrac{7\pi}{8}\sigma_x$ for, respectively, input 0 and 1 on the shared state $\ket{\Psi^+}=\tfrac{1}{\sqrt{2}}(\ket{01}+\ket{10})$.

Next, we consider $\vecPQ^\text{\tiny 90\%CHSH}$, which consists of a mixture of $\vecPQ^\text{\tiny CHSH}$ with the uniformly random distribution $\vecP_\one=\tfrac{1}{4}$. This allows us to gain some insight on how the regularization methods  fare for noisy quantum distributions, which are more readily accessible in the laboratory. Explicitly, the entries of this distribution are:
\begin{equation}
	P_\Q^\text{\tiny 90\%CHSH}(a,b|x,y)=\tfrac{1}{4}+\tfrac{9}{10}(-1)^{a+b+xy}\tfrac{\sqrt{2}}{8},
\end{equation}
which can be realized by performing the measurements mentioned above on the mixed state $\rho=\tfrac{9}{10}\proj{\Psi^+}+\tfrac{1}{10}\tfrac{\one}{4}$ where $\tfrac{\one}{4}$ is the maximally mixed two-qubit state.

For the purpose of device-independent property estimations, it is known~\cite{Moroder2013} that $\AQ_1$ {\em generally does not} provide a tight estimate of, e.g., the amount of negativity~\cite{Vidal2002} present~ in the system. An example of this is given by the quantum distribution $\vecPQ=\vecPQ^{\tau_{1.25}}$ which arises from both parties locally measuring the state $\ket{\psi}\simeq0.91\,\ket{00}+0.42\,\ket{11}$ in the basis of $\vec{n}_0\cdot\vec{\sigma}$ and $\vec{n}_1\cdot\vec{\sigma}$, where $\vec{n}_0\simeq(0.26, -0.97)$ and $\vec{n}_1\simeq-(0.87,0,0.49)$ are Bloch vectors. Explicitly, using the matrix representation given just below Eq.~\eqref{Eq:vecf:eg}, $\vecPQ^{\tau_{1.25}}$ reads as:
\begin{equation}\label{Eq:P}
\vecPQ^{\tau_{1.25}}=\left[\begin{array}{cc|cc}
\alpha_{00}  & \beta_{00} &   \alpha_{01}  & \beta_{01} \\
\gamma_{00}  & \epsilon_{00} &   \gamma_{01}  & \epsilon_{01} \\\hline
\alpha_{10}  & \beta_{10} &   \alpha_{11}  & \beta_{11} \\
\gamma_{10}  & \epsilon_{10} &   \gamma_{11}  & \epsilon_{11} \\
\end{array}\right]\,,
\end{equation}
where $\alpha_{00}\simeq0.00$, $\beta_{00}=\gamma_{00}\simeq0.01$, $\alpha_{01}=\alpha_{10}\simeq 0.00$, $\beta_{01}=\gamma_{10}\simeq 0.01$, $\beta_{10}=\gamma_{01}\simeq0.21$, $\alpha_{11}\simeq 0.05$, $\beta_{11}=\gamma_{11}\simeq 0.16$ and $\epsilon_{xy}=1-\alpha_{xy}-\beta_{xy}-\gamma_{xy}$ for all $x,y\in\{0,1\}$. Note that  $\vecPQ^{\tau_{1.25}}$ can be used to demonstrate {\em more nonlocality with less entanglement}~\cite{Junge2011,Liang2011,Vidick2011}, as was achieved in~\cite{Christensen2015}. In particular, being on the boundary of $\Q$, $\vecPQ^{\tau_{1.25}}$ {\em maximally} violates the $\tau=\tfrac{5}{4}$ version of the Bell inequality from~\cite{Liang2011}:
\begin{equation}\label{Ineq:Itau}
\begin{split}
	\I_{\tau}:=&\sum_{x,y} (-1)^{xy}P(0,0|x,y)\\
			-&\tau \sum_{a} \left[P(a,0|1,0)+P(0,a|0,1)\right]\stackrel{\L}{\le} 0
\end{split}
\end{equation}
which is provably~\cite{Christensen2015} satisfied by all finite-dimensional maximally entangled states whenever $\tfrac{1}{\sqrt{2}}+\tfrac{1}{2}\le \tau\le \tfrac{3}{2}$.

Finally, we also consider the correlation $\vecPQ^\text{\tiny MDL}$ of~\cite{Putz:NJP}: 
\begin{align}
	P_\Q^\text{\tiny MDL}(a,b|x,y)=&\tfrac{1}{12}(8ab+1)\delta_{x,0}\delta_{y,0} +\tfrac{1}{3}(1-\delta_{a,0}\delta_{b,0})\delta_{xy,1}\nonumber\\
				+&\tfrac{1}{6}(3ab+1)(1-\delta_{a,x}\delta_{b,y})\delta_{x\oplus y,1}.
\end{align}
which can be realized with both parties locally measuring $\sigma_x$  and $\sigma_z$ for, respectively, input 0 and 1 on the shared state $\ket{\Psi}=\tfrac{1}{\sqrt{3}}(\ket{01}+\ket{10}-\ket{11})$. $\vecP_\Q^\text{\tiny MDL}$ can be used to demonstrate the Hardy paradox~\cite{Hardy1993}, as well as a violation of measurement-dependent locality~\cite{Puetz2014} (MDL), e.g., via the following MDL inequality:
\begin{equation}\label{Ineq:MDL}
\begin{split}
	\I_\text{\tiny MDL} = &l P(0,0,0,0) - h \big[P(0,1,0,1)\\
	&+ P(1,0,1,0) + P(1,1,1,1)\big] \stackrel{\text{\tiny MDL}}{\le} 0,
\end{split}	
\end{equation}
where $h>l>0$. Note also that as opposed to $\vecPQ^\text{\tiny CHSH}$ which lies on the boundary of $\Q$, but strictly inside $\N$,  $\vecPQ^\text{\tiny MDL}$ lies on the boundary of both $\Q$ and $\N$. Experimental realizations of a correlation analogous to $\vecP_\Q^\text{\tiny MDL}$ have been achieved in~\cite{Aktas2015,Putz2016}.

\setcounter{figure}{2}

\subsection{Rate of convergence to the true distribution}
\label{App:Convergence}

\begin{figure*}
	\scalebox{0.24}{\includegraphics{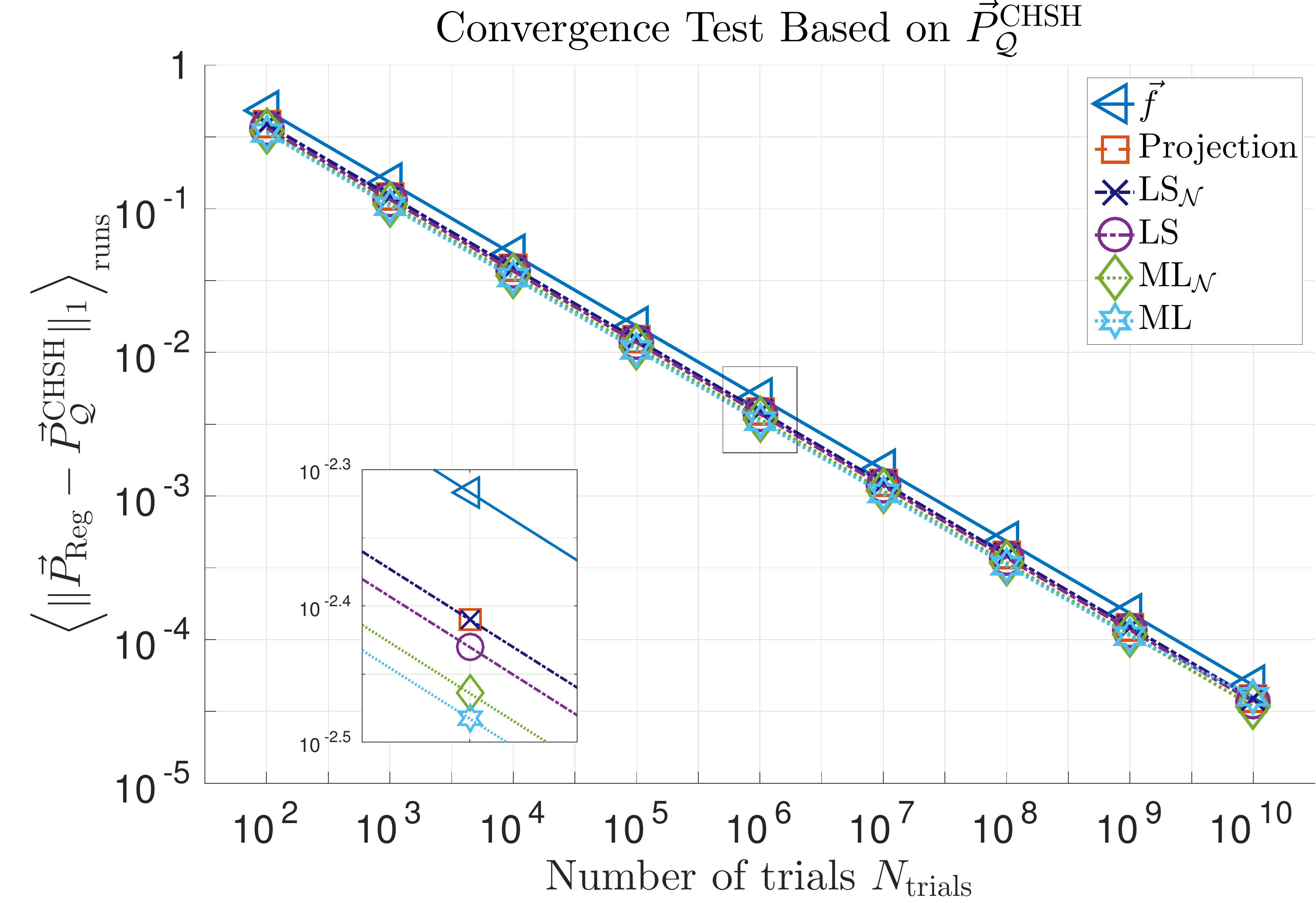}}\hspace{0.5cm}
	\scalebox{0.24}{\includegraphics{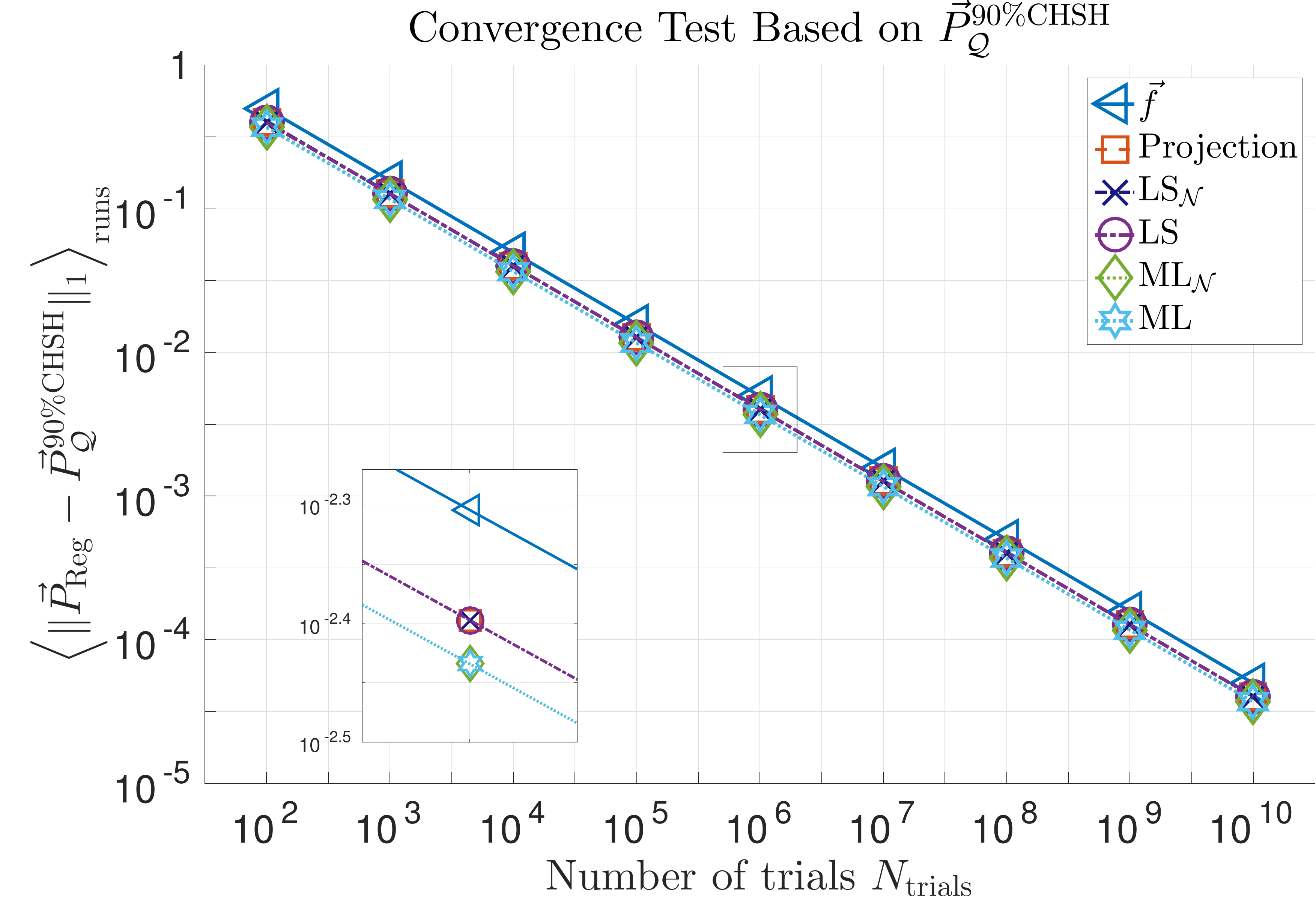}}\vspace{0.1cm}
	\vspace{0.1cm}
	\scalebox{0.24}{\includegraphics{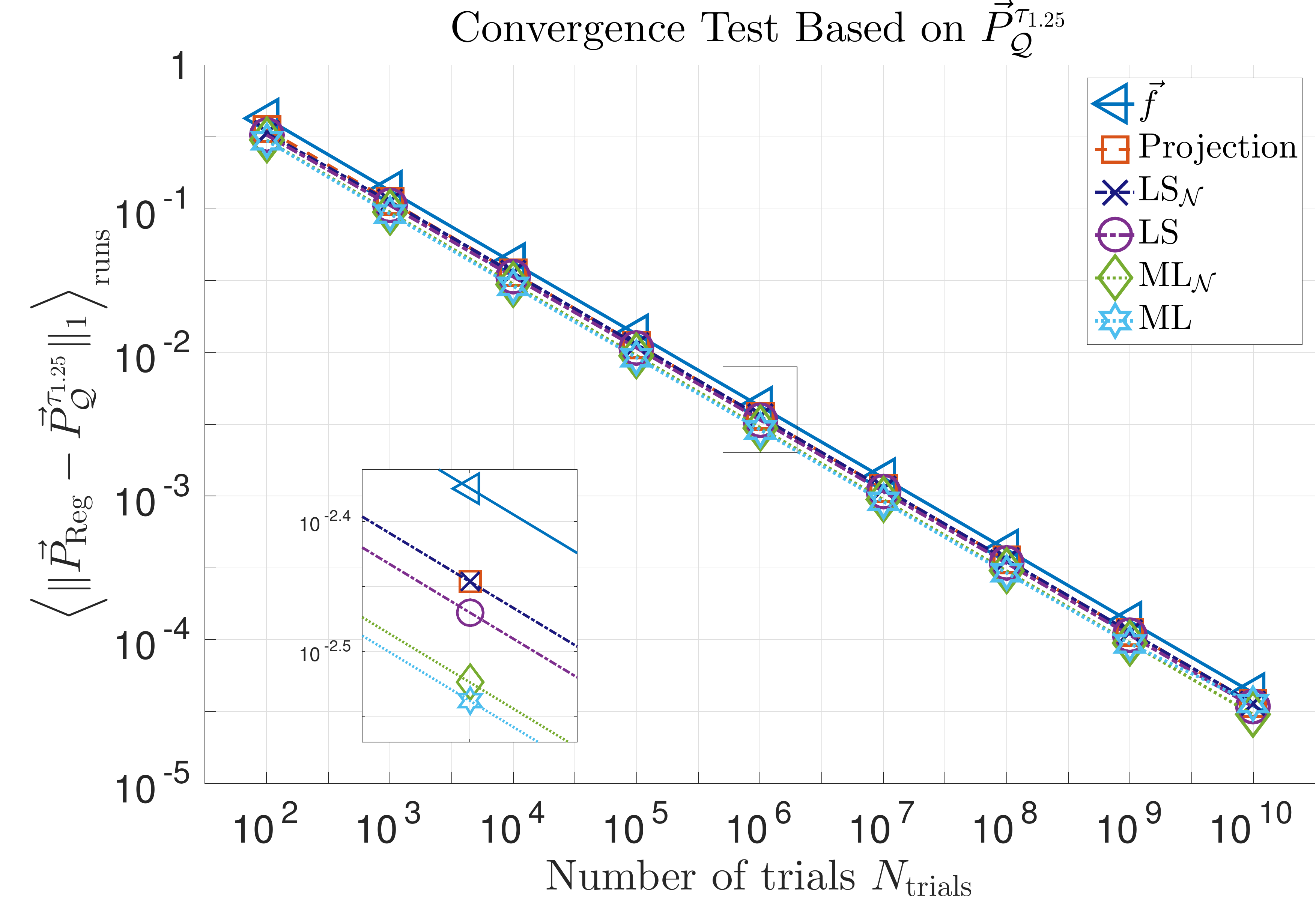}}\hspace{0.5cm}
	\scalebox{0.24}{\includegraphics{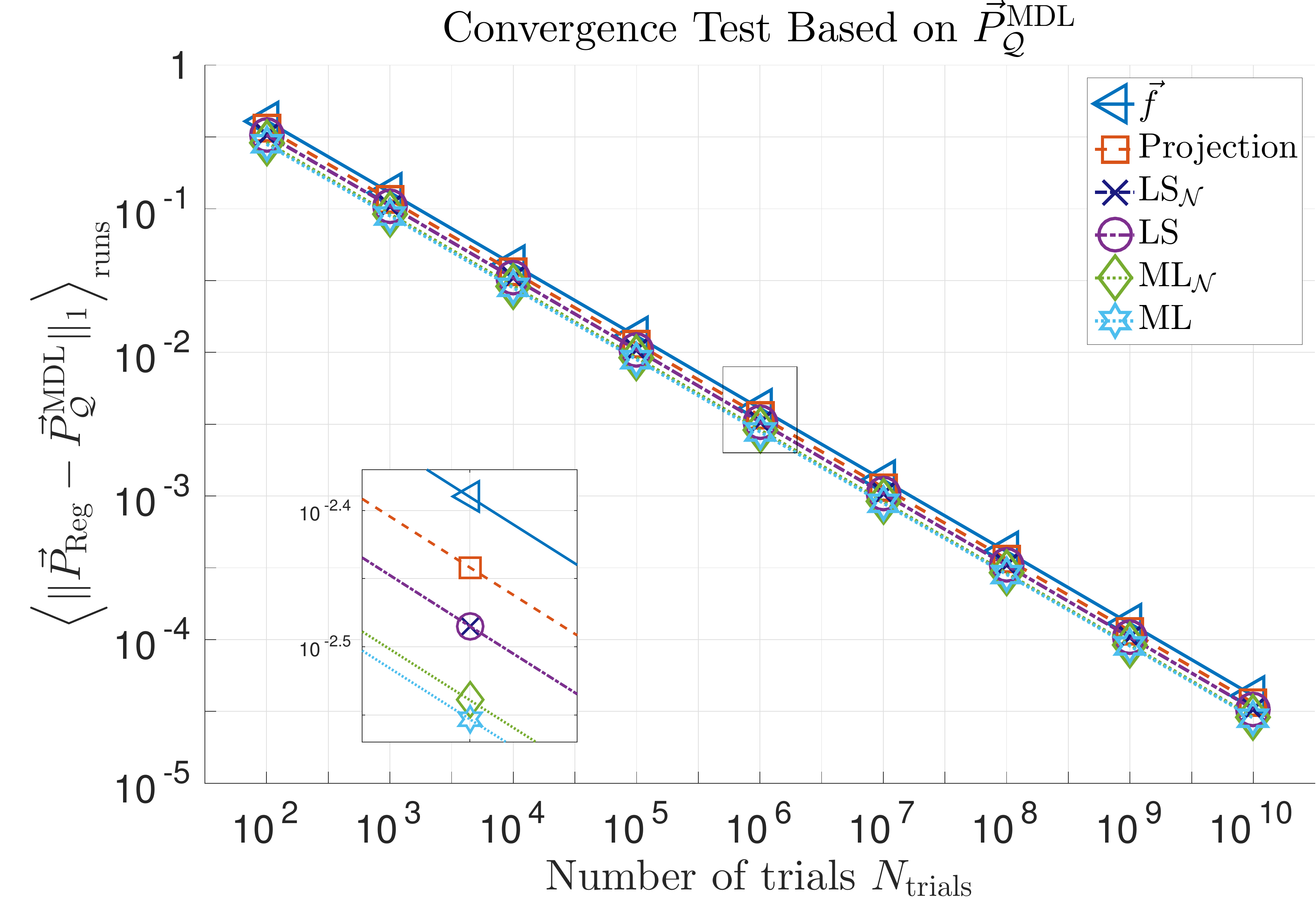}}
	\caption{\label{Fig:Convergence} (Color) Plots of the mean value of the 1-norm deviation  ${||\vecPreg(\vecf)-\vecPQ||_1}$, average over $10^4$ runs,  based on various $\vecPQ$ as a function of the number of trials $\Nt$ for the projection method (Projection), the modified least-square (LS) method with the nonsignaling polytope as the target set (LS$_\NSset$), the least-square method (LS), the modified ML method with a minimization to the nonsignaling polytope (ML$_\NSset$) and the ML method (ML). For details on the LS$_\NSset$ and the ML$_\NSset$ method, see Appendix~\ref{App:Other}. From top (left to right) to bottom (left to right), we have the plots based on $\vecPQ^\text{\tiny CHSH}$, the plots based on $\vecPQ^\text{\tiny 90\%CHSH}$, the plots based on $\vecPQ^{\tau_{1.25}}$, and the plots based on $\vecPQ^\text{\tiny MDL}$ (see Appendix~\ref{App:PQ} for details about these quantum distributions). In each subfigure, the corresponding inset shows a zoom-in view of the plots for $\Nt=10^6$.  }
\end{figure*}

We provide in Fig.~\ref{Fig:Convergence} the plots of the mean value of the 1-norm deviation ${||\vecPreg(\vecf)-\vecPQ||_1}$, between the regularized distribution $\vecPreg(\vecf)$ and the various $\vecPQ$ discussed above as a function of the number of trials $\Nt=10^2,10^3,\ldots,10^{10}$ for the regularization methods discussed in the main text and two additional regularization methods discussed in Appendix~\ref{App:Other}. For ease of comparison, we also include in each of these figures the corresponding plot for $\vecf$.

Notice that from some basic numerical fitting, one finds that for all these methods, the mean value of ${||\vecPreg(\vecf)-\vecPQ||_1}$, as with the mean value of ${||\vecf-\vecPQ||_1}$ diminishes at a rate of  $\tfrac{1}{\sqrt{\Nt}}$. In addition, since the 1-norm upper bounds all other $p$-norms with $p$ being an integer greater than or equal to 2, our results of 1-norm deviation also upper bound the deviation when measured in terms of other $p$ norms.

\subsection{Bias and mean squared errors of estimates}
\label{App:Bias}

\begin{figure*}
 	\scalebox{0.24}{\includegraphics{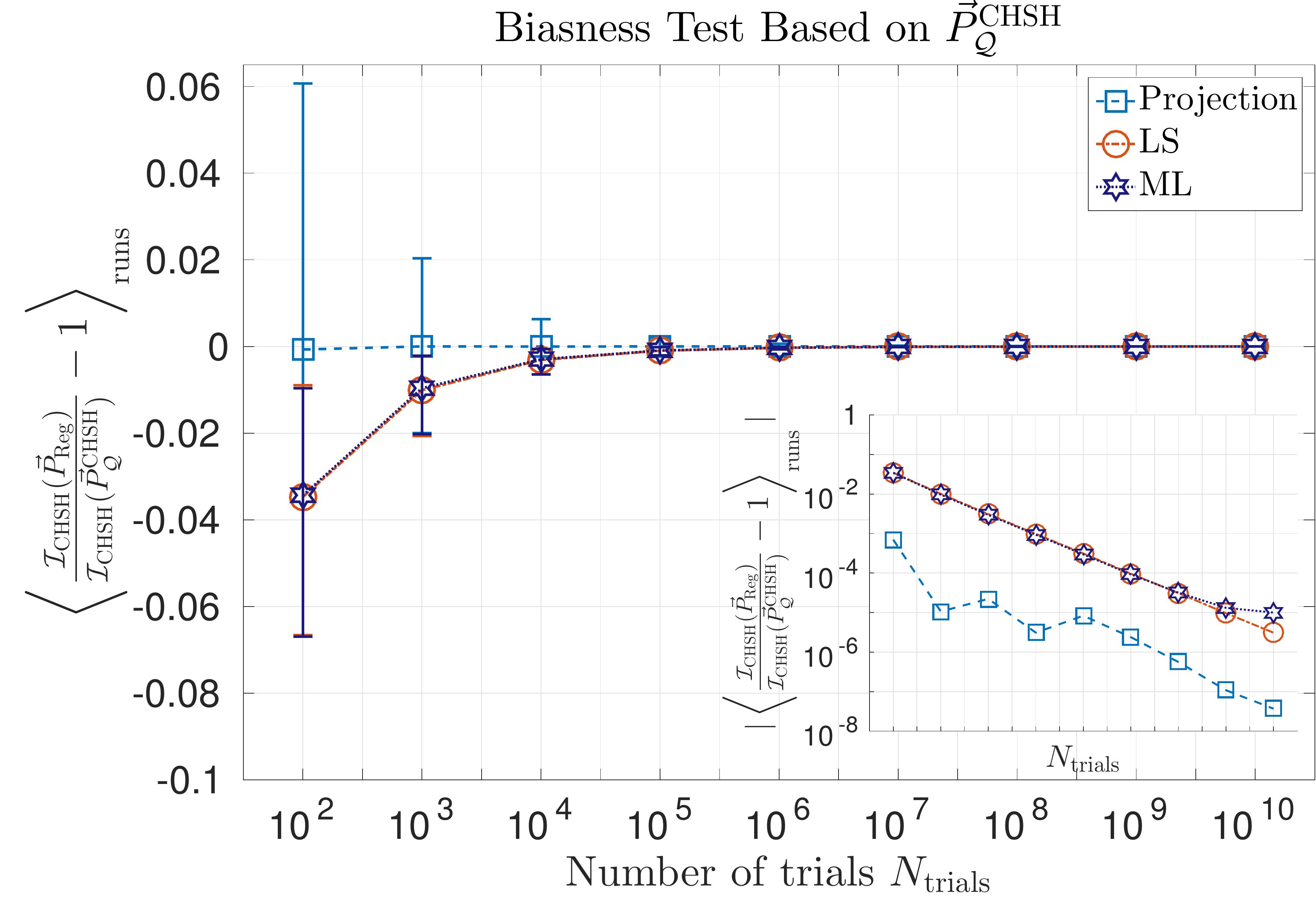}}\hspace{0.1cm}
	\scalebox{0.24}{\includegraphics{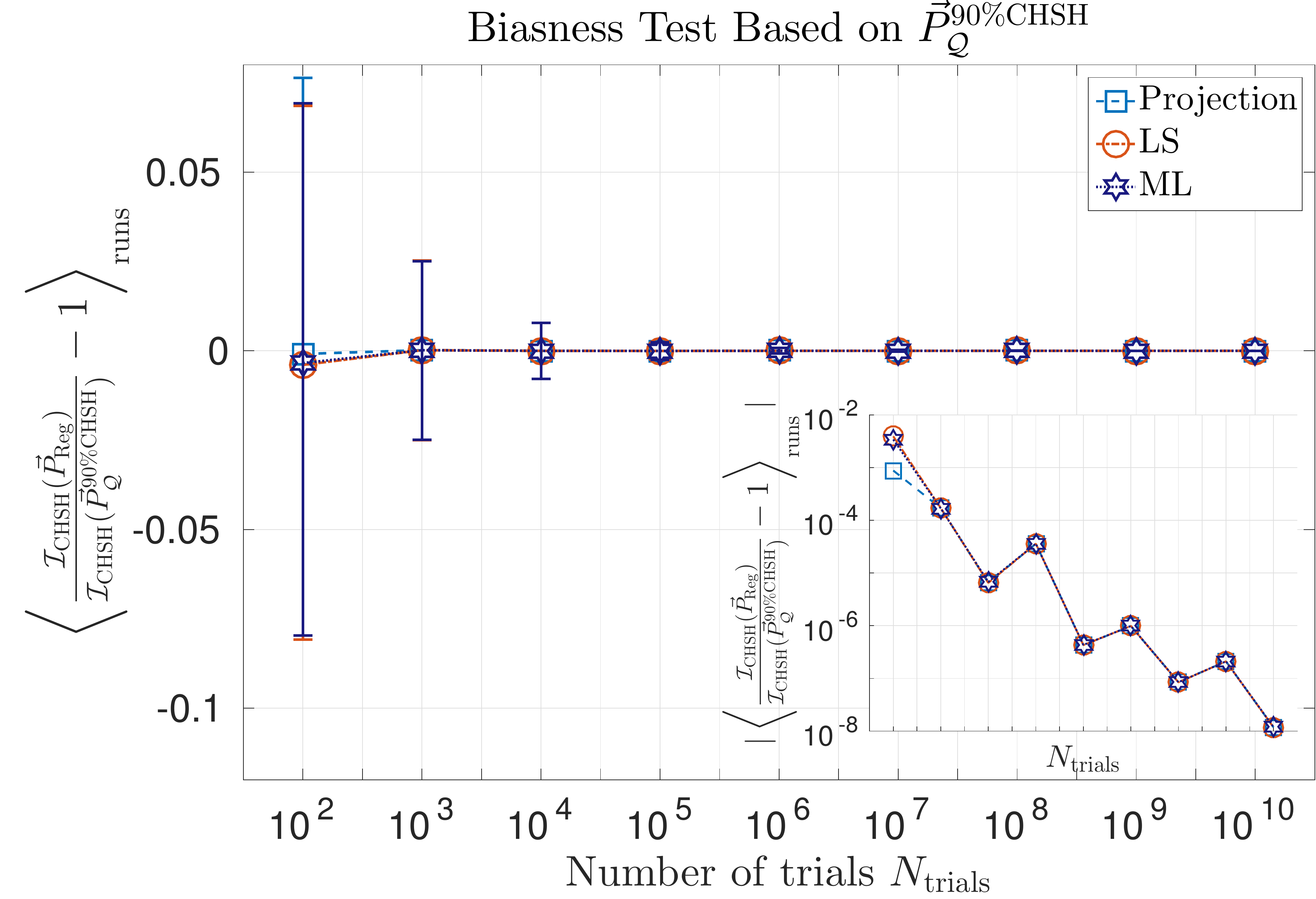}}\vspace{0.2cm}
	\scalebox{0.24}{\includegraphics{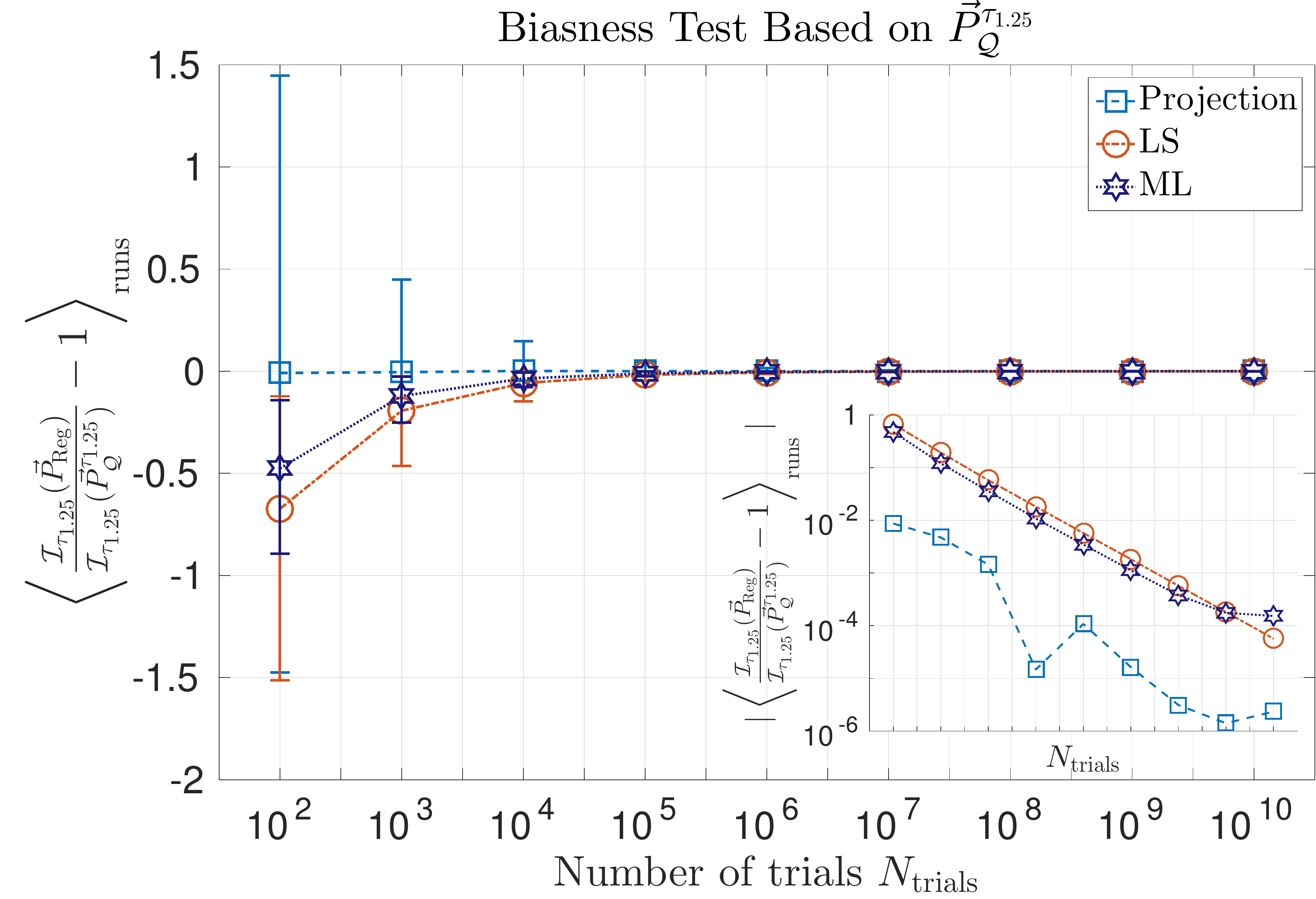}}
	\scalebox{0.24}{\includegraphics{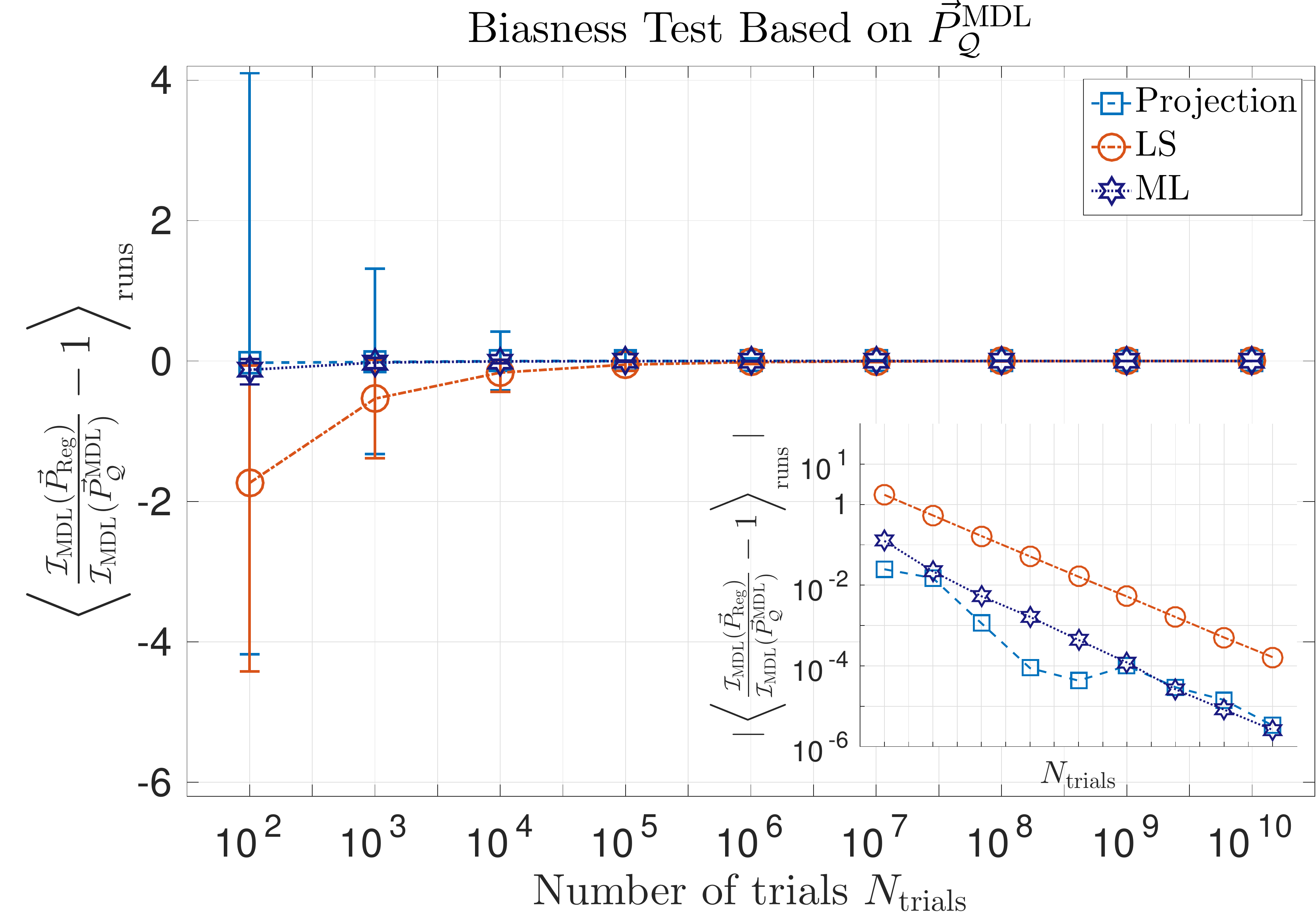}}
	\caption{\label{Fig:Bias} (Color) Plots of the mean value of the normalized Bell violations computed from $\vecPreg(\vecf)$ over $10^4$ relative frequencies $\vecf$ generated from various $\vecPQ$ as a function of the number of trials $\Nt$. From top (left to right) to bottom (left to right), we have  the plots based on $\vecPQ^\text{\tiny CHSH}$ and $\vecPQ^\text{90\%CHSH}$ in conjunction with the CHSH Bell inequality [Eq.~\eqref{Ineq:CHSH}], the plots based on $\vecPQ^{\tau_{1.25}}$ and the Bell inequality $\I_\tau$ with $\tau=1.25$ [Eq.~\eqref{Ineq:Itau}], as well as the plots based on $\vecPQ^\text{\tiny MDL}$ and the MDL inequality [Eq.~\eqref{Ineq:MDL}] with $P(x,y)=\tfrac{1}{4}$ for all $x,y$ and $l=0.1$, $h=1-3l=0.3$. For  details about the quantum distributions considered, see Appendix~\ref{App:PQ}. The lower and upper limit of each error bar mark, respectively, the 10\% and 90\% window of the spread of the values plotted. In each subfigure, the inset shows the corresponding log-log plot of the absolute value of the mean value.  }
\end{figure*}

\begin{figure*}
	\scalebox{0.24}{\includegraphics{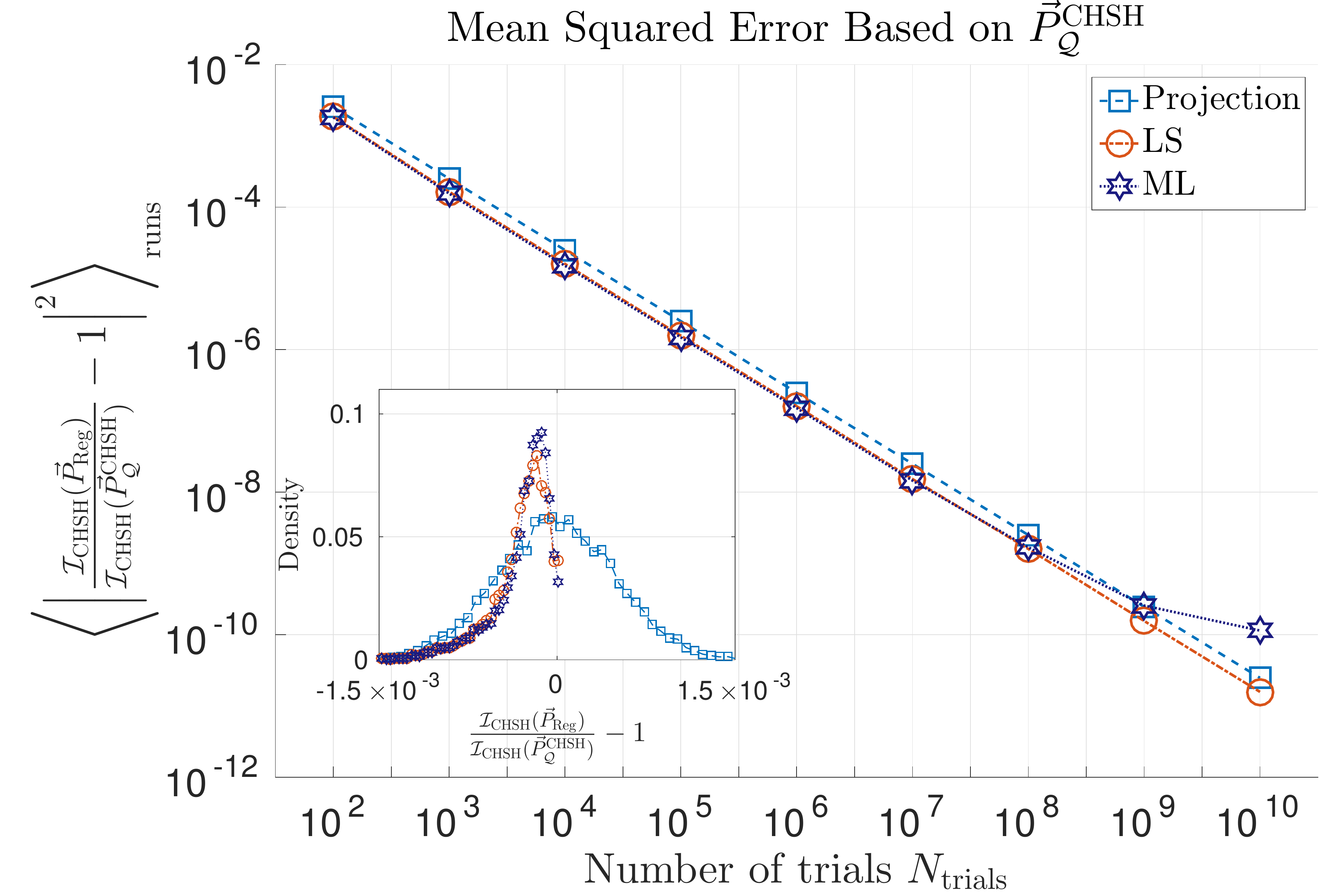}}
	\scalebox{0.24}{\includegraphics{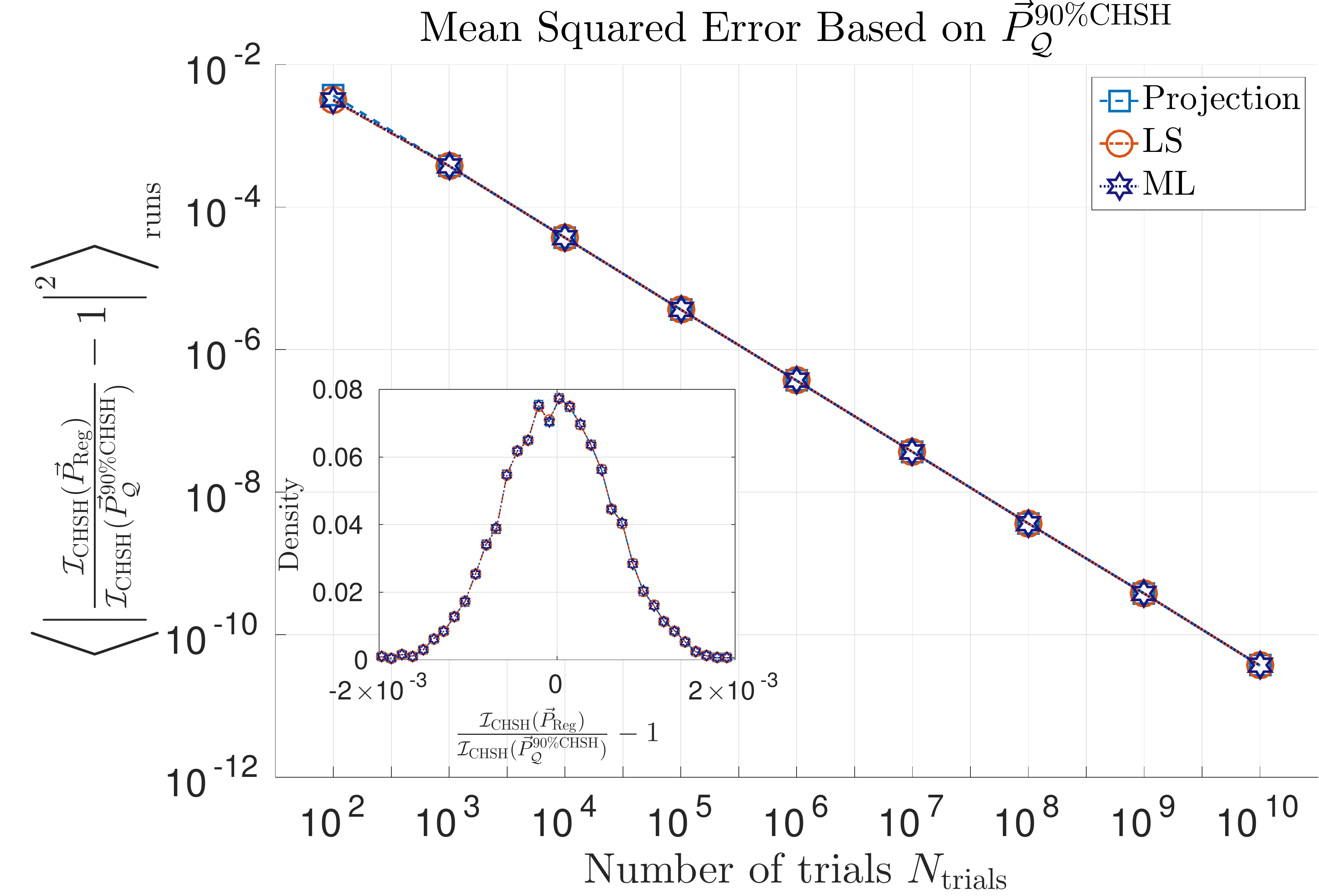}}\vspace{0.2cm}
	\scalebox{0.24}{\includegraphics{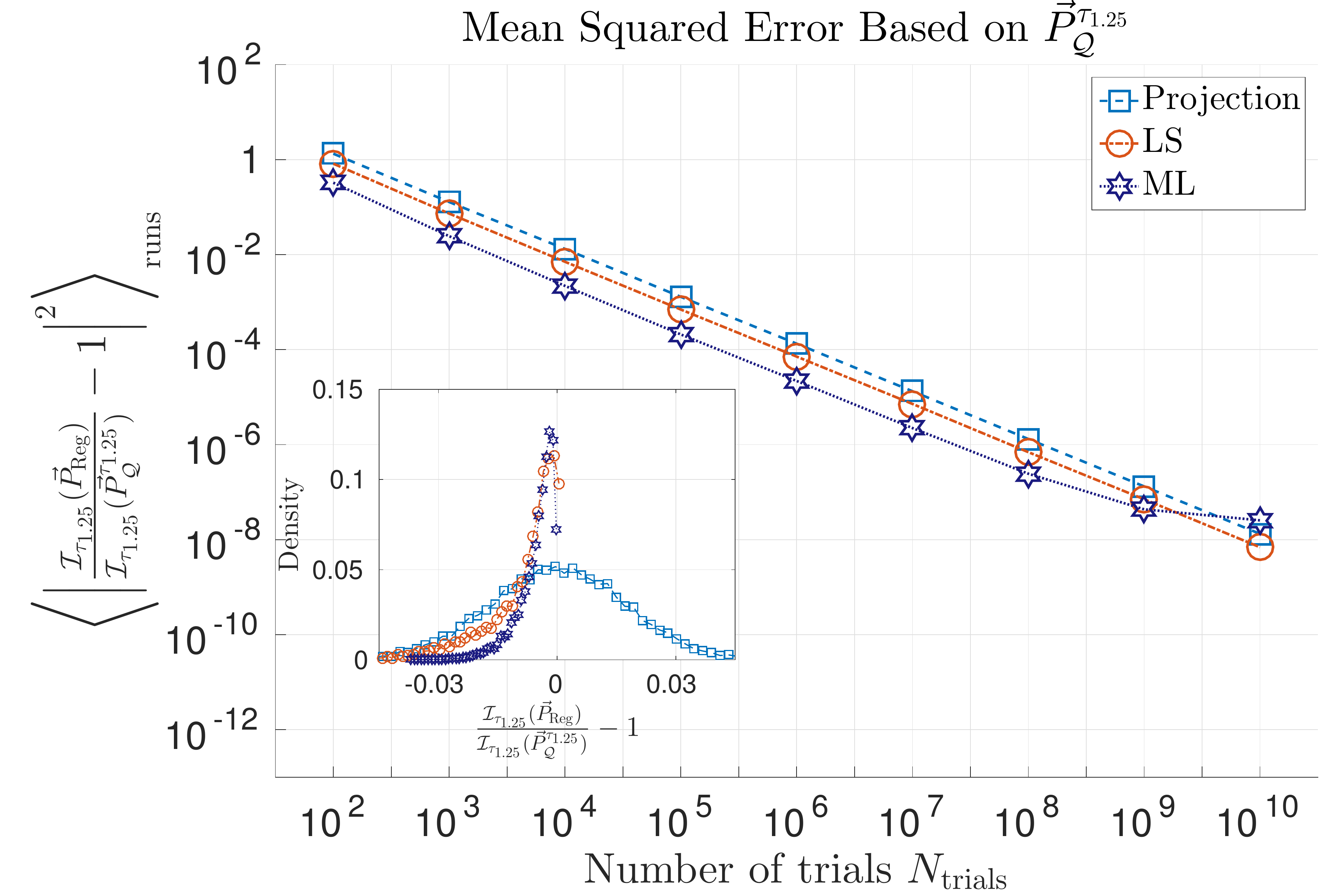}}
	\scalebox{0.24}{\includegraphics{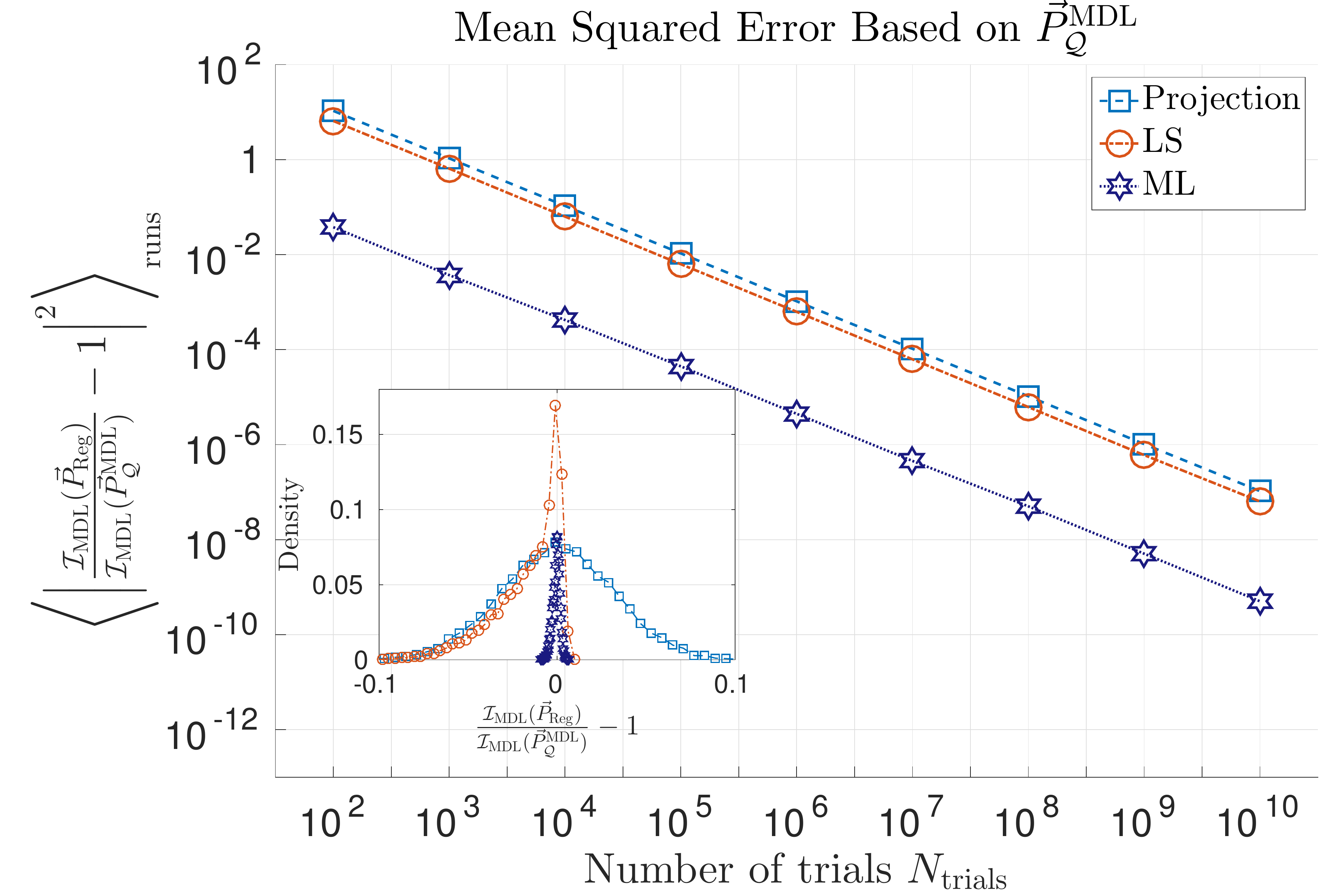}}
	\caption{\label{Fig:MSE} (Color) Plots of the mean squared error of the Bell violations corresponding to those described in Fig.~\ref{Fig:Bias}  over $10^4$ relative frequencies $\vecf$ generated from various $\vecPQ$ as a function of the number of trials $\Nt$. From top (left to right) to bottom (left to right), we have the plots based on $\vecPQ^\text{\tiny CHSH}$ and $\vecPQ^\text{90\%CHSH}$ in conjunction with the CHSH Bell inequality [Eq.~\eqref{Ineq:CHSH}], the plots based on $\vecPQ^{\tau_{1.25}}$ and the Bell inequality $\I_\tau$ with $\tau=1.25$ [Eq.~\eqref{Ineq:Itau}], as well as the plots based on $\vecPQ^\text{\tiny MDL}$ and the MDL inequality [Eq.~\eqref{Ineq:MDL}] with $P(x,y)=\tfrac{1}{4}$ for all $x,y$ and $l=0.1$, $h=1-3l=0.3$. For  details about the quantum distributions considered, see Appendix~\ref{App:PQ}. The lower and upper limit of each error bar mark, respectively, the 10\% and 90\% window of the spread of the values plotted. In each subfigure, the inset shows histograms of the normalized Bell violations for $\Nt=10^6$.  }
\end{figure*}

To gain further insight into the bias and the mean squared errors of various regularization methods, we provide in Fig.~\ref{Fig:Bias} and Fig.~\ref{Fig:MSE}  our simulation results for the mean Bell-inequality violation and the corresponding mean squared error based on the regularized distributions as a function of the number of trials $\Nt=10^2,10^3,\ldots,10^{10}$. Our results clearly suggest that the bias in the Bell value obtained from $\vecPpi$ is essentially negligible,\footnote{As the projection method involves a linear transformation of $\vecf$, its bias is in theory identically zero. The nonzero bias that we observe in this case arises from the fact that our numerical simulations involve only a finite number of samples ($10^4$).}  whereas that obtained from $\vecPnqa$ and $\vecPkl$---for extremal $\vecPQ$--- systematically underestimates (on average) the true value. However, as can be seen from the corresponding insets, such underestimations rapidly shrink with $\Nt$, diminishing at a rate of the order of $\tfrac{1}{\sqrt{\Nt}}$.  For the case of nonextremal $\vecPQ$, such as $\vecPQ^\text{\tiny90\%CHSH}$, we see that the bias present is essentially of the same order as that given by the projection method, which is basically negligible already for small $\Nt$. 

Similarly, as can be seen from Fig.~\ref{Fig:Bias}, the mean squared error rapidly decreases with $\Nt$ at a rate of the order of $\tfrac{1}{\Nt}$ in all the cases investigated. In particular, it is worth noting that for the case of $\vecPQ^\text{\tiny MDL}$, the mean squared error present for the ML method is approximately three orders of magnitude less than all those given by the other methods. This superiority of the ML method over the others is, to some extent, anticipated from the fact that the KL divergence is superior as a statistical distance over, e.g., the total variation distance in terms of discriminating probability distribution that contains zero entries, such as $\vecPQ^\text{\tiny MDL}$ (see page 28 of the preprint version of~\cite{vanDam2005} for a discussion).

In general, the bias nature of the few physical point estimators considered in Fig.~\ref{Fig:Bias} can be rigorously shown. See Appendix~\ref{App:Biased} for a proof.

\section{Proof of certain properties of the estimators}

\subsection{Uniqueness of estimators}
\label{App:Uniqueness}

We give here a proof that the output of certain regularization methods $\vecPreg(\vecf)$ is determined uniquely by $\vecf$. To this end, we first recall from Theorem 8.3 of~\cite{Beck2018} that the minimizer of a strictly convex function over a convex set is unique. To see that the LS method provides a unique estimate, let us first note that the Euclidean norm squared $(\|\vec{x}\|_2)^2$ is strictly convex in $\vec{x}$, since its Hessian  is two times the identity matrix. Moreover, $\min_{\vecP\in\AQ_\ell} (\|\vecf-\vecP\|_2)^2$ and $\min_{\vecP\in\AQ_\ell} \|\vecf-\vecP\|_2$ share exactly the same set of minimizer(s). Thus, the output of the LS method with $\vecPreg\in\C$ for any convex set $\C\subseteq\overline{\NSset}$ is necessarily unique. Likewise, since $-\log_2(x)$ is a strictly convex function of $x$,  the KL divergence from $\vecP\in\C$ to $\vecf$ is also strictly convex. In other words, the output of the regularization via both the LS method and the ML method is unique. In the main text, we have focused on $\C$ being some superset relaxation $\AQ_\ell$ of $\Q$, but this uniqueness clearly applies also to the nonsignaling polytope $\NSset$, thus allowing one to show the uniqueness of the LS and the ML estimators even when the target set is $\NSset$ (see Appendix~\ref{App:Other}).

\subsection{Bias of estimators}
\label{App:Biased}

We give here a proof of the biased nature of the physical estimator provided by the ML method and the LS method. To this end, we will first introduce the following definition for identifying correlations that may give rise to the same set of relative frequencies.
\begin{definition}\label{Dfn:StrictlyNonorth}
  Let $\vecP_1$ and $\vecP_2$ be two correlation vectors. We say that $\vecP_1$ and $\vecP_2$ are strictly nonorthogonal if for all input combinations, there is at least one combination of outcomes where the corresponding probability distributions of both $\vecP_1$ and $\vecP_2$ are nonvanishing. 
\end{definition}
Explicitly, in the bipartite scenario, the nonnegativity of probabilities means that $\vecP_1$ and $\vecP_2$ are strictly nonorthogonal if and only if $\sum_{a,b} P_1(a,b|x,y) P_2(a,b|x,y)>0$ for all $x,y$. The biased nature of $\vecPkl$ and $\vecPnqa$ is then an immediate consequence of Proposition~\ref{Prop:Biased}, which we prove as follows.

\begin{proof}
We will give a proof by contradiction. The steps follow closely those given in the proof of the main proposition of~\cite{Schwemmer2015}. For simplicity, we give the proof in a finite bipartite Bell scenario. The generalization to more complicated but still finite Bell scenarios follows analogously. Suppose a Bell experiment is carried out with $\vecP_{j}$ governing the underlying joint distribution of measurement outcomes. If the experiment involves only a finite number of trials $\Nt$,  and $\vecP_j$ is nondeterministic, one ends up with finite (nonsingleton) possibilities of relative frequencies $\vecf_i$ indexed by $i$. Let $\mathcal{F}_j=\{\vecf_i | q_{\vecP_{j}}(\vecf_i)>0\}$ be the set of all such relative frequencies obtainable from $\vecP_{j}$, where $q_{\vecP_{j}}(\vecf_i)\ge 0$ is the probability of observing the relative frequency $\vecf_i$ given that the underlying distribution is $\vecP_{j}$.

By assumption, one can find two strictly nonorthogonal extremal distributions of $\C$, say, $\vecP_{1}$ and $\vecP_{2}$, with $\vecP_{1}\neq \vecP_2$. Due to statistical fluctuations and the fact that $\vecP_{1}$ and $\vecP_{2}$ are strictly nonorthogonal, the set of relative frequencies simultaneously obtainable by both $\vecP_1$ and $\vecP_2$ is not empty, i.e., $\mathcal{F}_1\bigcap\mathcal{F}_2\neq\{\}$. Now, suppose that a regularization procedure gives an {\em unbiased} estimator. The expected distribution returned by the regularization scheme is then exactly the true distribution:
\begin{equation}
	\mathbb{E}_{\vecP_j}(\vecPreg)=\sum_{\vecf\in\F_j} q_{\vecP_j}(\vecf) \vecPreg(\vecf) = \vecP_j\quad\forall\, \vecP_j\in\C,
\end{equation}
where $\sum_j q_{\vecP_j}(\vecf)=1$. In particular, since $\vecP_1\in\C$, we have
\begin{equation}\label{Eq:Mixture}
	\sum_{\vecf\in\F_1} q_{\vecP_1}(\vecf) \vecPreg(\vecf) = \vecP_1.
\end{equation} 
By assumption, $\vecP_1$ is extremal in $\C$ and $\vecPreg(\vecf)\in\C$ for all $\vecf\in\F_1$. Thus, the unbiased nature of $\vecPreg(\vecf)$, see Eq.~\eqref{Eq:Mixture}, implies that $\vecPreg(\vecf) = \vecP_1$ for all $\vecf\in\F_1$.

Exactly the same line of reasonings can be applied to $\vecP_2$ to give $\vecPreg(\vecf) = \vecP_2$ for all $\vecf\in\F_2$. But as argued above, there exists $\vecf=\vecf_c$ such that $\vecf_c\in\F_1$ and $\vecf_c\in\F_2$. This implies that $\vecP_1=\vecPreg(\vecf_c) = \vecP_2$, which contradicts our assumption that $\vecP_1\neq\vecP_2$. Hence, given the premise that there exist extremal $\vecP_1$, $\vecP_2\in\C$ that are strictly nonorthogonal, the estimator $\vecPreg(\vecf)$, which is constrained to be a member of $\C$, must be biased.

\end{proof}  

To complete the proof that $\vecPkl$ and $\vecPnqa$ are biased, it suffices to choose $\vecP_1=\vecPQ^\CHSH$ and $\vecP_2$ as the deterministic point with $P_2(a,b|x,y)=\delta_{a,x}\delta_{b,y}$. Indeed, both these quantum distributions are known to be extremal in $\C=\Q$ (as well as any of its relaxation $\AQ_\ell$). Evidently, since no two extreme points of the nonsignaling polytope $\NSset$ in the simplest Bell scenario are strictly nonorthogonal, Proposition~\ref{Prop:Biased} cannot be used to show that a regularization method having $\NSset$ as its target is biased. Nonetheless, we conjecture that such estimators (which include all the other unique estimators given in Table~\ref{Tab:MethodProperties}) must also be biased.

\section{Device-independent negativity estimation, optimized witnesses and Bell inequality}
\label{App:Neg}

In the scenarios we studied, the sets $\AQ_1$, $\AQ_2$, ..., $\AQ_4$ correspond to outer approximations of the quantum set $\Q$ using moment matrices of (local) level $\ell = 1, 2, \ldots$ (see~\cite{Moroder2013,Vallins2017} for details). As shown in~\cite{Moroder2013}, the same relaxations provide a way to lower bound the amount of entanglement present in a quantum system. We discuss here a geometrical formulation of their method. Let $\Q^{\le \nu}$ be the set of all correlations $\vecPQ$ obtained from Born's rule using a state $\rho$ of maximal negativity $\nu$. Following~\cite{Moroder2013}, we write $\AQ_\ell^{\le \nu} \supseteq \Q^{\le \nu}$ as the corresponding semidefinite relaxation of level $\ell$. Then, given correlations $\vecPQ$ and approximation level $\ell$, a lower bound on the negativity of $\rho$ can be obtained by computing the supremum $\nu$ such that $\vecPQ \not\in \AQ_\ell^{\le\nu}$.

What are the requirements on the distributions used as input to the negativity estimation algorithm? We first remark that $\AQ_\ell^{\le \nu}$ is a subset of the nonsignaling set, so the negativity estimation algorithm can {\em never} be performed on the relative frequency $\vecf$, but should instead employ one of the regularized distributions $\vecPreg(\vecf)$. Secondly, the set $\AQ_\ell^{\le\nu}$ is a subset of $\AQ_\ell$, so that $\vecf$ should be regularized to a target set $\AQ_{\ell'}$ with $\ell' \ge \ell$. Otherwise, we run the risk of having $\vecPreg \in \AQ_{\ell'} \setminus \AQ_{\ell}$ and the semidefinite program will turn out to be infeasible. In practice, this discrepancy could happen even by regularizing to $\ell' = \ell$ due to insufficient numerical precision, in which case a small amount of white noise can be added to restore feasibility.  For both $\ell=1,2$, this typical fraction of white noise is found to be of the order of 10$^{-8}$ or less.

In our negativity estimation shown in Fig.~\ref{Fig:Negativity:Mean} of the main text, we have employed $\AQ_2$ as our approximation to the quantum set $\Q$. This may seem unnecessary as $\AQ_1$~\cite{Navascues2015} is known to be a pretty good approximation of the quantum set. However, as in the case of obtaining a negativity bound from the CHSH Bell-inequality violation~\cite{Moroder2013}, one finds that $\AQ_1$ generally does not provide a tight negativity bound. This becomes evident by plotting $\AQ_\ell$ and $\AQ_\ell^{\le \nu}$ around the distribution $\vecPQ^{\tiny \tau_{1.25}}$, as shown  in Fig.~\ref{Fig:hierarchies}. There, we see that although $\vecPQ^{\tiny \tau_{1.25}}$ lies visually on the boundary of both $\AQ_1$ and $\AQ_2$, the approximations $\AQ_\ell^{\le\nu}$ of the bounded negativity sets differ significantly, to the point that, in the considered slice, the approximation level $\ell=1$ is unable to certify any point with a negativity higher than $0.3$.

\begin{figure}[h!]
\includegraphics{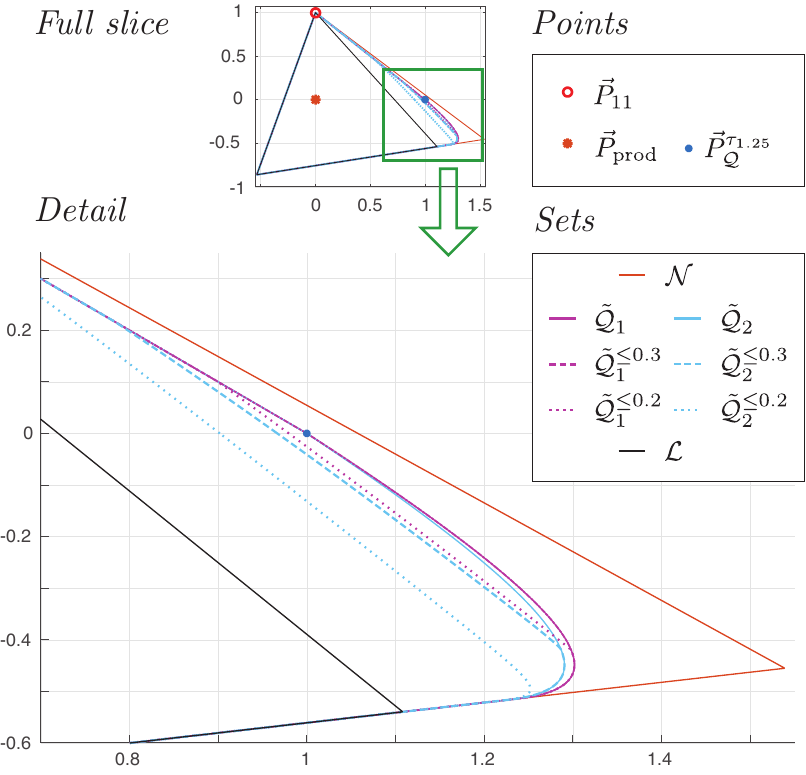}
  \caption{\label{Fig:hierarchies} (Color) A two-dimensional slice in the space of distributions. The point $\vecPQ^{\tiny \tau_{1.25}}$ is the distribution considered in the study of negativity. The deterministic point $\vec{P}_{11}$ corresponds to $P(1,1|x,y) = 1$ for all $x,y$. The product distribution $P_\text{prod}(a,b|x,y) = P_{\mathcal{Q}}^{\tiny \tau_{1.25}}(a|x) P_{\mathcal{Q}}^{\tiny \tau_{1.25}}(b|y)$ is obtained from the marginals of $P_{\mathcal{Q}}^{\tiny \tau_{1.25}}$. In black and red respectively, we have the boundary of the local and the nonsignaling polytopes $\L$ and $\N$. We also plot, for outer approximations of level $\ell = 1$ (blue) and $\ell = 2$ (magenta), the approximated quantum set $\AQ_\ell$, and sets of upper-bounded negativity $\AQ_\ell^{\le \nu}$ for $\nu = 0.2$ and $0.3$.
  }
\end{figure}

Let us also clarify here the connection between the optimized device-independent witness obtained by running a DI algorithm and an ordinary Bell inequality. By solving an SDP analogous to that given in~\cite{Moroder2013} with a regularized distribution $\vecPreg(\vecf)$, the dual of the SDP gives an optimized device-independent witness (a Bell-like inequality) of the form:\footnote{Extracting this witness follows a procedure analogous to that given in Appendix A of~\cite{Schwarz2016}.}
\begin{equation}\label{Eq:OptimizedNegW}
	\sum_{a,b,x,y} \beta^{xy}_{ab} P(a,b|x,y)\stackrel{\text{N}(\rho)\le \alpha}{\le } \S_\alpha,
\end{equation}
where $\alpha\ge 0$ and $\S_\alpha$ are some {\em fixed} numbers, and $\sum_{a,b,x,y} \beta^{xy}_{ab} P_\text{\tiny Reg}(a,b|x,y)> \S_\alpha$. 

Inequality~\eqref{Eq:OptimizedNegW} is an (optimized) device-independent negativity witness in the sense that for any (nonsignaling) $\vecP$, if it gives rise to the left-hand-side of Eq.~\eqref{Eq:OptimizedNegW} a value that is within the limit allowed by quantum theory and is greater than $\S_\alpha$, then whatever quantum state gives rise to $\vecP$ must have a negativity greater than $\alpha$. Thus, assuming that $\vecPreg(\vecf)$ is a reasonable estimate of the underlying distribution, we see that it gives a negativity estimate of the underlying state that is at least $\alpha$.

Alternatively, ~\cite{Moroder2013} provides a way to estimate negativity, in a device-independent manner, by starting from a given Bell-inequality violation: given any vector of coefficients $\beta_{ab}^{xy}$, one can compute, for each value $\S = \sum_{abxy} \beta_{ab}^{xy}P(a,b|x,y)$ a corresponding lower bound $\alpha$ on  the negativity. For example, when one uses the coefficients of the CHSH inequality as $\beta_{ab}^{xy}$, we obtain the negativity bound mentioned in the main text: N$(\rho)\ge \tfrac{\S_\CHSH-2}{4\sqrt{2}-4}$ for $\S_\CHSH\in[2,2\sqrt{2}]$, which gives a monotonous relation between $\alpha$ and $\S$. Note that only when $\S_\CHSH$ is greater than the local bound 2 one can hope to obtain a nonzero lower bound on negativity. Also, $\S_\CHSH> 2\sqrt{2}$ is not quantum realizable and thus cannot be used for the device-independent estimation of negativity.

In the same spirit, any vector of coefficients $\beta_{ab}^{xy}$ can be converted to a family of witnesses. In this process, one can even reuse the optimized Bell inequality given in ~\eqref{Eq:OptimizedNegW}, providing a way to interpret values of $\S \ge \S_\alpha$ for which the coefficients $\beta_{ab}^{xy}$ were originally computed.

\section{Some other plausible regularization methods and their properties}
\label{App:Other}

Here, we  discuss a few other plausible regularization methods with a target set $\C\in\{\Q,\NSset\}$. Properties of these methods are summarized in Table~\ref{Tab:MethodProperties}.

\begin{table*}
\centering
\begin{tabular}{|c|c|c|c|c|}
\hline
Method            & Output set $\mathcal{R}$                       & Optimization & Unique? & Unbiased?   \\ \hline\hline
Projection        &  \NSspace                           & -    & \checkmark    & \checkmark               \\ \hline
NNA$_1$      & \multirow{3}{*}{\NSset}            & LP    & \xmark      & -               \\ 
LS$_\NSset$      &                                                & SOCP  & \checkmark    & \xmark                \\ 
NNA$_\infty$ &                                                & LP    & \xmark      & -                \\ \hline\hline
NQA$_1$      & \multirow{2}{*}{$\Qset$} & SDP   & \xmark      & -              \\
LS      & \multirow{2}{*}{(approximated by $\AQ_\ell \supseteq \Qset$)} & SDP   & $\checkmark$    & \xmark               \\
NQA$_\infty$ &  & SDP   & \xmark      & -                \\ \hline\hline
ML$_\NSset$           &  \NSset                             & CP    & $\checkmark$      & \xmark                \\ \hline
\multirow{2}{*}{ML}           & $\Qset$                              & \multirow{2}{*}{CP}    & \multirow{2}{*}{$\checkmark$}      & \multirow{2}{*}{\xmark}               \\ 
          &  (approximated by $\AQ_\ell \supseteq \Qset$)                             &     &       &                \\ \hline\hline
\end{tabular}
\caption{\label{Tab:MethodProperties}A summary of the various regularization methods discussed  and some of their key properties. Note that all these considered methods are invariant under relabelings, in the sense that for any relabeling operation $M$, we have $M\vecPreg(\vecf)=\vecPreg(M\vecf)$. In the event that the output of the regularization method is unique, we indicate in the last column if the method provides (in general) unbiased estimation of Bell violation.}
\end{table*}

\subsection{Nearest quantum or nonsignaling approximation via $p$-norms}

Obviously, one can  regularize a given $\vecf$ to $\Q$ (resp. $\N$) by determining the nearest quantum (resp. nonsignaling) approximation NQA (resp. NNA) of $\vecf$ to $\Q$ (resp. $\N$) with a metric induced by any of the $p$-norms:
\begin{subequations}
\begin{gather}\label{Eq:NQAp}
	\vecP_\text{\tiny NQA$_p$}(\vecf)=\argmin_{\vecP\in\Q} ||\vecf-\vecP||_p,\\
	\vecP_\text{\tiny NNA$_p$}(\vecf)=\argmin_{\vecP\in\NSset} ||\vecf-\vecP||_p.	\label{Eq:NNAp}
\end{gather}
\end{subequations}
When $p=2$,  NQA$_2$ gives the LS method described in Appendix~\ref{App:NQA}. By the same token, we shall refer to NNA$_2$ of Eq.~\eqref{Eq:NNAp} as LS$_\N$. We can rewrite it as the following second-order cone program (SOCP)~\cite{Boyd2004Book}:
\begin{equation}\label{Eq:NNA_2}
\begin{split}
	&\min \qquad s\\
	&\text{s.t.}\quad \| \vecf-\vecP \|_2\le s\\
	&\qquad\  d_i\le \vec{c}_i\cdot\vecP\quad\forall\quad i=1,2,\ldots,m
\end{split}
\end{equation}
where inequalities in the last line are positivity constraints used to define the nonsignaling polytope. Hence, the LS$_\N$ regularization can be also efficiently computed using an SOCP solver. Note that the inequality constraints in the last line of Eq.~\eqref{Eq:NNA_2} can be replaced by imposing the nonsignaling constraints of Eq.~\eqref{Eq:NS}. The regularization method of LS$_\N$ is then evidently a least-square minimization problem with linear equality constraints. This regularization method has previously been implemented in~\cite{Cavalcanti:2015aa} as part of their data analysis.

For the case of $p=1$ (previously considered in~\cite{Schwarz2016}) or $p=\infty$, the optimization problem of Eqs.~\eqref{Eq:NQAp} and ~\eqref{Eq:NNAp} can be cast, respectively, as a semidefinite program and a linear program. For both values of $p$, one can find easily an example of $\vecf$ where some $\vecP_\text{\tiny NQA$_p$}(\vecf)$ [$\vecP_\text{\tiny NNA$_p$}(\vecf)$] is Bell-inequality violating but some other is not. For explicit examples, see the Supplemental Material~\cite{MatFile}. Similarly, the variant of NNA$_1$ employed in~\cite{Bernhard2014} is known to give nonunique estimators~\cite{Schwarz:Private}.

\subsection{Minimizing the KL divergence to $\NSset$}

As with the LS$_\N$ method, one can also consider minimizing the KL divergence from the nonsignaling polytope $\NSset$ to some given relative frequency $\vecf$ 
\begin{equation}\label{Eq:KL-NP}
	\min_{\vecP\in\NSset}\,\, \sum_{a, b, x, y} f(x,y) f(a,b|x,y) \log_2 \left[ \frac{f(a,b|x,y)}{P(a,b|x,y)} \right].
\end{equation}
We shall refer to the corresponding regularization method as the ML$_\NSset$ method. Note that, as with the ML method, performing the ML$_\NSset$ regularization method amounts to solving a conic program, which can be achieved using, e.g., the SCS solver.

Although not explicitly discussed as a regularization method, the ML$_\NSset$ method has been employed in~\cite{Zhang2011} and was noted to help in the analysis of the hypothesis testing of local causality, as discussed in Appendix 2 of~\cite{Zhang2011} [see Eqs. (A1) and (A2) therein]. Conceivably, the ML method may help in the analysis of Bell tests further.

\end{document}